\documentclass[fleqn,10pt]{wlscirep}
\usepackage[utf8]{inputenc}
\usepackage[T1]{fontenc}
\usepackage{graphicx} % Required for inserting images
\usepackage{subcaption}
\usepackage{multirow}
\usepackage{booktabs}

\usepackage{enumitem}

\usepackage{algorithm}
\usepackage{algpseudocode}
\usepackage{amsmath}
\usepackage{amsthm}
\newtheorem{theorem}{Theorem}

\newtheorem{definition}{Definition}
\usepackage{braket}
\usepackage[compat=0.4]{yquant}
\yquantdefinebox{null}[inner sep=0pt]{}
\useyquantlanguage{groups}

\usepackage{hyperref}

\usepackage{xcolor}

\usepackage[normalem]{ulem}

\usepackage{todonotes}

\usepackage{tikz}
\usetikzlibrary{calc,patterns,shapes.geometric}
\usepackage{siunitx}
\usepackage{pgfplots}
\usepgfplotslibrary{statistics}
\pgfplotsset{compat=newest}
\usepgfplotslibrary{units}
\pgfdeclareplotmark{circle*}{%
    \node[fill=none,draw=none,shape=circle,inner sep=1pt] {*};  % Circle with an embedded *
}
\definecolor{idealcolor}{RGB}{31,119,180} % Blue
\definecolor{LNNcolor}{RGB}{255,127,14} % Orange
\definecolor{ZZcolor}{RGB}{44,160,44} % Green
\definecolor{EnolaDycolor}{RGB}{214,39,40} % Red
\definecolor{EnolaScolor}{RGB}{148,103,189} % Purple
\definecolor{DasAtomcolor}{RGB}{128,128,128} % Gray
\definecolor{Atomiquecolor}{RGB}{70,130,180} % Steel Blue
\definecolor{darkblue}{rgb}{0.0, 0.0, 0.55}
\definecolor{darkorange}{rgb}{1.0, 0.55, 0.0}
\definecolor{darkgreen}{rgb}{0.0, 0.5, 0.0}
\tikzset{
    path2Q/.style={draw=darkblue, fill=darkblue, pattern=horizontal lines, pattern color=darkblue, thick},
    pathDec/.style={draw=darkorange, fill=darkorange, pattern=crosshatch, pattern color=darkorange, thick},
    pathTran/.style={draw=darkgreen, fill=darkgreen, pattern=grid, pattern color=darkgreen,thick},
    totalerrorstyle/.style={semithick, black, mark=*, mark size=1, mark options={solid}},
    idealPlot/.style={color=idealcolor,thick,mark=none,},
    EnolaDyPlot/.style = {color=EnolaDycolor, mark=diamond*, only marks,},
    LNNPlot/.style = {color=LNNcolor, mark=*, only marks,},
    ZZPlot/.style = {color=ZZcolor, mark=o, mark options={solid},,only marks,},
    EnolaSPlot/.style = {color=EnolaScolor, mark=circle*, only marks,},
    DasAtomPlot/.style={draw=DasAtomcolor, thick, mark=square*, only marks, fill=DasAtomcolor, pattern=north west lines, pattern color=DasAtomcolor},
    AtomiquePlot/.style={draw=Atomiquecolor, thick, mark=triangle*, only marks, fill=Atomiquecolor, pattern=vertical lines, pattern color=Atomiquecolor},
}
\pgfplotsset{
    CompareAxis/.style={
        tick align=outside,
        tick pos=left,
        x grid style={white!50.1960784313725!black},
        xlabel={Number of Qubits},
        xmin=5, xmax=50,
        xtick style={black,line width=2pt},
        y dir=reverse,
        y grid style={color=gray!30},
        ylabel style={at={(axis description cs:-0.2,0.5)},anchor=south},
        ylabel={Fidelity, log scale},
        ymajorgrids,
        ymin=0, ymax=45,
        xlabel style={font=\huge},
        ylabel style={font=\huge},
        ytick={0,5,10,15,20,25,30,35,40,45}, 
        ytick style={black,line width=2pt},
        yticklabels={0,\(\displaystyle 10^{-5}\),\(\displaystyle 10^{-10}\),\(\displaystyle 10^{-15}\),\(\displaystyle 10^{-20}\),\(\displaystyle 10^{-25}\),\(\displaystyle 10^{-30}\),\(\displaystyle 10^{-35}\),\(\displaystyle 10^{-40}\),\(\displaystyle 10^{-45}\)},
    },
}
\tikzset{
    control atom/.style={circle, draw=blue, pattern=north east lines, pattern color=blue!50, thick, minimum size=0.03cm},
    target atom/.style={circle, draw=red, pattern=grid, pattern color=red!50, thick, minimum size=0.03cm},
    dashed node/.style={circle, draw, dashed, fill=gray!10, minimum size=0.03cm}
}

\usepackage{nameref}
\usepackage{makecell}

\usepackage{amssymb,amsmath,amsthm,amsfonts}
\theoremstyle{definition}
\newtheorem{example}{Example}

\theoremstyle{theorem}
\newtheorem{lemma}{Lemma}

\title{Optimal Compilation Strategies for QFT Circuits in Neutral-Atom Quantum Computing}

\author[1]{Dingchao Gao}
\author[2]{Yongming Li}
\author[1]{Shenggang Ying}
\author[3,*]{Sanjiang Li}
\affil[1]{Key Laboratory of System Software (Chinese Academy of Sciences) and State Key Laboratory of Computer Science, Institute of Software, Chinese Academy of Sciences}
\affil[2]{School of Mathematics and Statistics, Shaanxi Normal University, Xi'an, 710062, China}
\affil[3]{Centre for Quantum Software and Information (QSI), Faculty of Engineering and Information Technology, University of Technology Sydney, NSW 2007, Australia}

\affil[*]{sanjiang.li@uts.edu.au}

\begin{abstract}
Neutral-atom quantum computing (NAQC) offers distinct advantages such as dynamic qubit reconfigurability, long coherence times, and high gate fidelities, making it a promising platform for scalable quantum computing. Despite these strengths, efficiently implementing quantum circuits like the Quantum Fourier Transform (QFT) remains a significant challenge due to atom movement overheads and connectivity constraints. This paper introduces optimal compilation strategies tailored to QFT circuits and NAQC systems, addressing these challenges for both linear and grid-like architectures. By minimizing atom movements, the proposed methods achieve theoretical lower bounds in atom movements while preserving high circuit fidelity. Comparative evaluations against state-of-the-art compilers demonstrate the superior performance of the proposed methods. These methods could serve as benchmarks for evaluating the performance of NAQC compilers.
\end{abstract}
\begin{document}

\flushbottom
\maketitle
\thispagestyle{empty}

\section{Introduction}
Quantum computing has emerged as a transformative technology capable of solving complex problems in cryptography \cite{shor:focs94}, database search \cite{grover1996fast}, chemical simulations \cite{peruzzo2014variational}, and machine learning \cite{schuld2015introduction}. The rapid advancement of quantum computing research has led to the development of multiple hardware platforms, each with unique strengths and challenges. Among these, superconducting qubits have gained prominence due to their robust control schemes and compatibility with conventional microelectronics~\cite{arute2019quantum}. In the past several years, \emph{neutral-atom quantum computing} (NAQC) has attracted significant attention for its inherent advantages, including scalability, high qubit connectivity, long coherence times, and superior gate fidelity~\cite{henriet2020quantum, Bluvstein_2022, radnaev2024universal}.

In superconducting quantum devices, such as IBM's heavy‐hex architecture, qubits are sparsely arranged to minimize crosstalk. This sparse arrangement, however, requires additional routing steps (mainly through inserting SWAP gates) to facilitate operations between distant qubits. In contrast, NAQC devices, often arranged on a 2D grid, enable two-qubit gates between any two qubits by shuttling them closer. However, longer shuttling distances can increase noise levels. Thus, careful circuit compilation is critical  on both platforms to achieve high circuit fidelity.

The \emph{Quantum Fourier Transform} (QFT) is a foundational component in many quantum algorithms, such as Shor's factoring algorithm \cite{shor:focs94}, phase estimation \cite{kitaev1995phase_estimation}, quantum simulation~\cite{Lloyd_simulator_1996},  amplitude amplification \cite{Brassard_2002}, and the HHL algorithm for solving systems of linear equations \cite{Harrow_hhl_2009}. Despite its importance, implementing QFT efficiently on current quantum hardware represents a significant challenge. The QFT's all-to-all interaction pattern requires careful consideration of hardware-specific constraints, particularly in systems with limited qubit connectivity. Due to its core role in quantum computing, QFT optimization has been extensively studied in the literature, see \cite{maslov+:physreva07,takahashi2007quantum,nam2020approximate,zhang2021time,jin_quantum_2023,Park2023ReducingCC,gao_linear_2024,baumer2024quantum}.

In the realm of superconducting quantum devices, Maslov \cite{maslov+:physreva07} proposed a linear-depth transformation for QFT circuits on the linear nearest-neighbor (LNN) architecture, where qubits lie along a single path. Modern superconducting devices often feature more complex connectivity, making it challenging to identify a single Hamiltonian path that visits all qubits. Building on this foundational work, Jin \emph{et al.}~\cite{jin_quantum_2023} introduced efficient mapping techniques tailored to general 2D grid and IBM heavy-hex architectures. Focused on IBM's heavy-hex architecture, Gao \emph{et al.}~\cite{gao_linear_2024}  proposed specifically optimized linear-depth transformations. 

For NAQC devices, costly SWAP gates can be entirely replaced with atom movement (cf.~\cite{Schmid24qst}). While several compilation algorithms have been proposed \cite{tan2024compilation,wang2024atomique,Wang23Q-Pilot,huang2020predicting}, they fail to produce optimal transformations. Furthermore, there is a lack of benchmarking tools to evaluate how close these generated transformations are to the optimal solutions. 

In this work, we address these challenges by introducing optimal compilation strategies tailored to QFT circuits and NAQC platforms. Our contributions are threefold:
\begin{itemize}
    \item We propose a Linear Path strategy that achieves theoretical lower bounds in movement counts for linear architectures.
    \item We extend this approach to 2D grid-like architectures via a Zigzag Path strategy, which preserves gate parallelism while minimizing movement overhead.
    \item We demonstrate that our methods outperform state-of-the-art compilation techniques, such as Enola~\cite{tan2024compilation} and Atomique~\cite{wang2024atomique}, exponentially in qubit count and in terms of movement efficiency and overall fidelity.
\end{itemize}

Our strategies leverage the unique capabilities of NAQC, such as dynamic qubit reconfigurability and global Rydberg illumination, to optimize circuit execution while minimizing atom movements. Furthermore, we extend our methods to MaxCut QAOA circuits \cite{farhi2014quantum} with reduced connectivity, demonstrating their versatility for a wide range of quantum algorithms beyond fully connected QFT circuits. By addressing the inefficiencies of atom movement and connectivity constraints, this work advances the implementation of practical quantum algorithms on NAQC platforms. Our methods could serve as benchmarks for evaluating quantum circuit compilation algorithms and inspire the development of new, more efficient approaches.

The remainder of this paper is organized as follows. Section~\ref{sec:background} reviews NAQC and discusses the challenges of compilation in NAQC. Section~\ref{sec:qft} revisits QFT circuits and their optimal transformation in linear superconducting architectures. Section~\ref{sec:opt} details our proposed method for NAQC platforms, including strategies for both linear and grid-like architectures. Section~\ref{sec:eva} presents experimental evaluations and comparisons with state-of-the-art algorithms. Finally, Section~\ref{sec:con} concludes with a summary of our findings and their implications for scalable quantum computing.

%\cyan{The `section \ ref' command doesn't work in this template because it only supports the `/\ section{}' format. Therefore, I used phrases like `first' and `then' to connect the `sections' within the paragraph.}

%\cyan{The remainder of this paper is organized as follows. First, we review NAQC and discuss the challenges of compilation in NAQC. Then, we revisit QFT circuits and their optimal transformation in linear superconducting architectures. Next, we detail our proposed method for NAQC platforms, including strategies for both linear and grid-like architectures. Afterwards, we present experimental evaluations and comparisons with state-of-the-art algorithms. Finally, we conclude with a summary of our findings \red{and their implications for scalable quantum computing.}}

\begin{figure}[htbp]
\centering
%\begin{minipage}[b]{0.3\columnwidth}
%\centering
%\resizebox{\columnwidth}{!}{\input{figure/heavy-hex}}
%\subcaption{}
% \subcaption{Architecture graphs of IBM's heavy-hex superconducting device}
%\label{fig:ibm-hex}
%\end{minipage}
%\hspace{0.1\columnwidth}
%\begin{minipage}[b]{0.18\columnwidth}
%\centering
%\resizebox{\columnwidth}{!}{\input{figure/grid}}
%\subcaption{}
% \subcaption{Architecture graphs of IBM's heavy-hex superconducting devicean NAQC device.}
%\label{fig:grid}
%\end{minipage}

%\vspace{0.5cm}

%\begin{minipage}{\columnwidth}
    \resizebox{\columnwidth}{!}{% \documentclass{standalone}
% \usepackage{tikz}
% \usetikzlibrary{calc}
% \input{style}
% \begin{document}

\begin{tikzpicture}
    % Helper styles for better readability
    \tikzstyle{unoccupied}=[fill=gray!20, circle, minimum size=4pt, inner sep=0pt]
    \tikzstyle{occupied}=[fill=black, circle, minimum size=4pt, inner sep=0pt]
    \tikzstyle{aod_circle}=[fill={rgb,255:red,200;green,80;blue,80}, circle, minimum size=4pt, inner sep=0pt]
    \tikzstyle{aod_line}=[color={rgb,255:red,200;green,80;blue,80},thick, dotted]
    \tikzstyle{rb_circle}=[fill={rgb,255:red,160;green,210;blue,200}, opacity=0.6]
    
    \begin{scope}
        % Grid
        % \draw[ultra thin, gray, dashed] (0, 0) grid (4, 4);
        % Axes
        \draw[->, thick] (0,0) -- (4, 0) node[right] {$x$};
        \draw[->, thick] (0,0) -- (0, 4) node[above] {$y$};

        % x and y labels
        % Labels for Axes
        \node[below] at (1,0) {$d$};
        \node[left] at (0,1) {$d$};
        \foreach \x in {2,3} \node[below] at (\x,0) {$\x d$};
        \foreach \y in {2,3} \node[left] at (0,\y) {$\y d$};
        % X-axis ticks
        \foreach \x in {1,2,3} {
            \draw (\x,0) -- (\x,-0.1);
        }
        
        % Y-axis ticks
        \foreach \y in {1,2,3} {
            \draw (0,\y) -- (-0.1,\y);
        }
        
        % AOD activation (red dashed lines)
        \draw[aod_line] (0.2, 1) -- (3.4, 1) node[right] {$x_0$};
        \draw[aod_line] (0.2, 3) -- (3.4, 3) node[right] {$x_1$};
        \draw[aod_line] (1, 0.2) -- (1, 3.4) node[above] {$y_0$};
        \draw[aod_line] (2, 0.2) -- (2, 3.4) node[above] {$y_1$};
        \draw[aod_line,->,thick,solid] (1,3.3) -- (1.8,3.3);
        \draw[aod_line,->,thick,solid] (2,3.3) -- (2.8,3.3);
        \draw[aod_line,->,thick,solid] (3.3,1) -- (3.3,1.8);
    
        % Unoccupied nodes - dashed circles
        \foreach \x in {1,2,3} {
            \foreach \y in {1,2,3} {
                \node[unoccupied] at (\x, \y) {};
            }
        }
        
        % SLM occupied nodes (filled circles)
        \foreach [count=\i from 0] \x/\y in {1/3, 2/1, 1/1,3/2,2/3} {
            \node[occupied] at (\x, \y) {};
        }
        
        % AOD occupied nodes (circled filled circles)
        \foreach \x/\y in {1/3, 2/1} {
            \node[aod_circle] at (\x, \y) {};
        }

        % Title
        \node[below] at ($(2,-0.6)$) {\textbf{(a)}};
    \end{scope}

    \begin{scope}[shift={(5,0)}]
        % Grid
        % \draw[ultra thin, gray, dashed] (0, 0) grid (4, 4);
        % Axes
        % \draw[->, thick] (0,0) -- (4, 0) node[right] {$x$};
        % \draw[->, thick] (0,0) -- (0, 4) node[above] {$y$};

        % % x and y labels
        % % Labels for Axes
        % \node[below] at (1,0) {$d$};
        % \node[left] at (0,1) {$d$};
        % \foreach \x in {2,3} \node[below] at (\x,0) {$\x d$};
        % \foreach \y in {2,3} \node[left] at (0,\y) {$\y d$};
        % % X-axis ticks
        % \foreach \x in {1,2,3} {
        %     \draw (\x,0) -- (\x,-0.1);
        % }
        
        % % Y-axis ticks
        % \foreach \y in {1,2,3} {
        %     \draw (0,\y) -- (-0.1,\y);
        % }
        
        % AOD activation (red dashed lines)
        % \draw[aod_line] (0.2, 1) -- (3.4, 1) node[right] {$x_0$};
        % \draw[aod_line] (0.2, 3) -- (3.4, 3) node[right] {$x_1$};
        % \draw[aod_line] (1, 0.2) -- (1, 3.4) node[above] {$y_0$};
        % \draw[aod_line] (2, 0.2) -- (2, 3.4) node[above] {$y_1$};
        % \draw[aod_line,->,thick,solid] (1,3.3) -- (1.8,3.3);
        % \draw[aod_line,->,thick,solid] (3.3,1) -- (3.3,1.8);
    
        % Unoccupied nodes - dashed circles
        \foreach \x in {1,2,3} {
            \foreach \y in {1,2,3} {
                \node[unoccupied] at (\x, \y) {};
            }
        }
        
        % SLM occupied nodes (filled circles)
        \foreach [count=\i from 0] \x/\y in {1/1,3/2,2/3} {
            \node[occupied] at (\x, \y) {};
        }
        
        % AOD occupied nodes (circled filled circles)
        \foreach \x/\y in {1.8/3, 2.8/2} {
            \node[aod_circle] at (\x, \y) {};
        }
        
        \draw[aod_line,->] (1.1,3) -- (1.7,3);
        \draw[aod_line,->] (2.1,1.1) -- (2.7,1.9);
        % Title
        \node[below] at ($(2,-0.6)$) {\textbf{(b)}};
    \end{scope}
    \begin{scope}[shift={(10,0)}]
        % Grid
        % \draw[ultra thin, gray, dashed] (0, 0) grid (4, 4);
        % Axes
        % \draw[->, thick] (0,0) -- (4, 0) node[right] {$x$};
        % \draw[->, thick] (0,0) -- (0, 4) node[above] {$y$};

        % % x and y labels
        % % Labels for Axes
        % \node[below] at (1,0) {$d$};
        % \node[left] at (0,1) {$d$};
        % \foreach \x in {2,3} \node[below] at (\x,0) {$\x d$};
        % \foreach \y in {2,3} \node[left] at (0,\y) {$\y d$};
        % % X-axis ticks
        % \foreach \x in {1,2,3} {
        %     \draw (\x,0) -- (\x,-0.1);
        % }
        
        % % Y-axis ticks
        % \foreach \y in {1,2,3} {
        %     \draw (0,\y) -- (-0.1,\y);
        % }
        
        % AOD activation (red dashed lines)
        % \draw[aod_line] (0.2, 1) -- (3.4, 1) node[right] {$x_0$};
        % \draw[aod_line] (0.2, 3) -- (3.4, 3) node[right] {$x_1$};
        % \draw[aod_line] (1, 0.2) -- (1, 3.4) node[above] {$y_0$};
        % \draw[aod_line] (2, 0.2) -- (2, 3.4) node[above] {$y_1$};
        % \draw[aod_line,->,thick,solid] (1,3.3) -- (1.8,3.3);
        % \draw[aod_line,->,thick,solid] (3.3,1) -- (3.3,1.8);
    
        % Unoccupied nodes - dashed circles
        \foreach \x in {1,2,3} {
            \foreach \y in {1,2,3} {
                \node[unoccupied] at (\x, \y) {};
            }
        }
        
        % SLM occupied nodes (filled circles)
        \foreach [count=\i from 0] \x/\y in {1/1,3/2,2/3,1.8/3, 2.8/2} {
            \fill[rb_circle] (\x, \y) circle (0.4);
        }
        \foreach [count=\i from 0] \x/\y in {1/1,3/2,2/3,1.8/3, 2.8/2} {
            \node[occupied] at (\x, \y) {};
        }
        % Title
        \node[below] at ($(2,-0.6)$) {\textbf{(c)}};
    \end{scope}
    \begin{scope}[shift={(15,0)}]
        % Grid
        % \draw[ultra thin, gray, dashed] (0, 0) grid (4, 4);
        % Axes
        % \draw[->, thick] (0,0) -- (4, 0) node[right] {$x$};
        % \draw[->, thick] (0,0) -- (0, 4) node[above] {$y$};

        % % x and y labels
        % % Labels for Axes
        % \node[below] at (1,0) {$d$};
        % \node[left] at (0,1) {$d$};
        % \foreach \x in {2,3} \node[below] at (\x,0) {$\x d$};
        % \foreach \y in {2,3} \node[left] at (0,\y) {$\y d$};
        % % X-axis ticks
        % \foreach \x in {1,2,3} {
        %     \draw (\x,0) -- (\x,-0.1);
        % }
        
        % % Y-axis ticks
        % \foreach \y in {1,2,3} {
        %     \draw (0,\y) -- (-0.1,\y);
        % }
        
        % AOD activation (red dashed lines)
        % \draw[aod_line] (0.2, 1) -- (3.4, 1) node[right] {$x_0$};
        % \draw[aod_line] (0.2, 3) -- (3.4, 3) node[right] {$x_1$};
        % \draw[aod_line] (1, 0.2) -- (1, 3.4) node[above] {$y_0$};
        % \draw[aod_line] (2, 0.2) -- (2, 3.4) node[above] {$y_1$};
        % \draw[aod_line,->,thick,solid] (1,3.3) -- (1.8,3.3);
        % \draw[aod_line,->,thick,solid] (3.3,1) -- (3.3,1.8);
    
        % Unoccupied nodes - dashed circles
        \foreach \x in {1,2,3} {
            \foreach \y in {1,2,3} {
                \node[unoccupied] at (\x, \y) {};
            }
        }
        
        % SLM occupied nodes (filled circles)
        % \foreach [count=\i from 0] \x/\y in {1/1,3/2,2/3,1.8/3, 2.8/2} {
        %     \fill[rb_circle] (\x, \y) circle (0.4);
        % }
        \foreach [count=\i from 0] \x/\y in {1/1,3/2,2/3,} {
            \node[occupied] at (\x, \y) {};
        }
        \draw[aod_line] (1.8, 0.2) -- (1.8, 3.4) node[above] {};
        \draw[aod_line] (2.8, 0.2) -- (2.8, 3.4) node[above] {};
        
        \draw[aod_line] (0.2, 3) -- (3.4, 3) node[above] {};
        \draw[aod_line] (0.2, 2) -- (3.4, 2) node[above] {};
        \draw[aod_line,->,thick,solid] (1.8,3.3) -- (1,3.3);
        \draw[aod_line,->,thick,solid] (2.8,2.3) -- (2,2.3);
        \foreach \x/\y in {1.8/3, 2.8/2} {
            \node[aod_circle] at (\x, \y) {};
        }
        % Title
        \node[below] at ($(2,-0.6)$) {\textbf{(d)}};
    \end{scope}
    % % Legend
    \node[anchor=south] at ($(current bounding box.north) + (0, 0.5em)$) {
        \begin{tikzpicture} % Nested TikZ scope for the legend
            \node[unoccupied, label=right:{Optical trap}] at (0, 0) {};
            \node[occupied,   label=right:{SLM atom}]     at (3, 0) {};
            \node[aod_circle, label=right:{AOD atom}]     at (6, 0) {};
            \fill[rb_circle] (9, 0.05) circle (0.2);
            \node at (9, -0.05) [label=right:{Rydberg laser}] {};
        \end{tikzpicture}
    };
\end{tikzpicture}

% \end{document}
}
    %\subcaption{}
    %\label{fig:fpqa}
%\end{minipage}

\caption{
%Illustrations of different quantum computing architectures.
%(\subref{fig:ibm-hex})~IBM's heavy-hex superconducting layout.(\subref{fig:grid})~a grid-based NAQC layout. 
%(\subref{fig:fpqa})~a 
Illustration of atom movement and Rydberg illumination in a $3\times 3$ grid DPQA architecture: (a) Atoms are initially positioned at separate grid points (i.e., sites);  (b) During the ``meet'' step, atoms moves so that each interacting pair is co-located at the same grid point;  (b) A global Rydberg laser applies \(CZ\) gates simultaneously to all co-located atom pairs; (d) In the  ``separate'' step, atom pairs moves apart, ensuring they occupy different grid points, though not necessarily their initial positions.}
%\label{fig:architectures}
\label{fig:fpqa}
\end{figure}
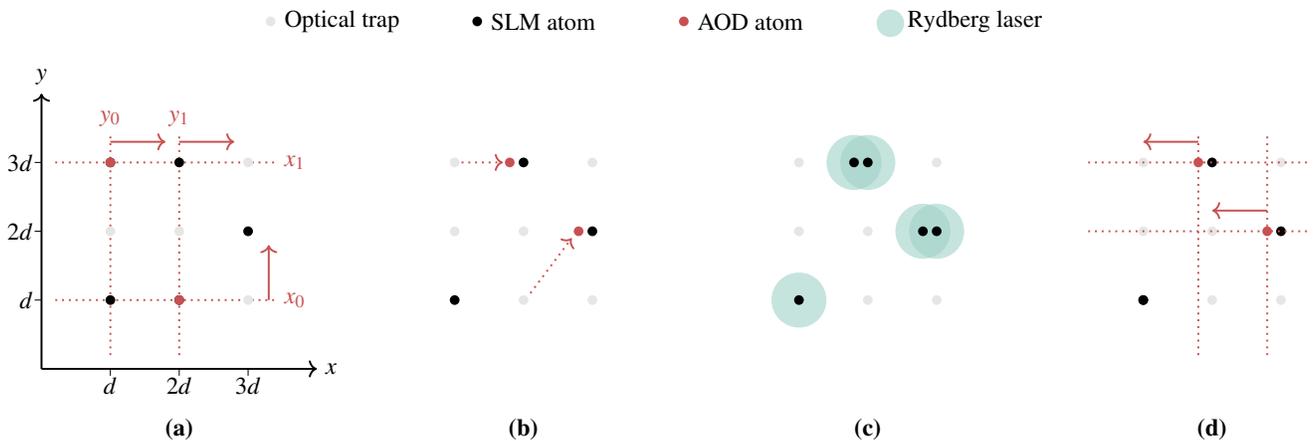

\section{Background}\label{sec:background}
In this section, we provide background on quantum computing and NAQC technology and examine the compilation challenges specific to NAQC.

\subsection{Quantum Computing Basics}
Quantum computing leverages quantum mechanical phenomena such as superposition and entanglement to solve problems that are intractable for classical computers. A quantum circuit, the computational model used in most quantum algorithms, consists of qubits manipulated by single- and multi-qubit gates. These gates form a sequence of unitary operations that evolve the quantum state towards the desired solution. However, the efficiency of quantum circuit execution is heavily influenced by the architecture and physical implementation of qubits. Hardware constraints such as limited qubit connectivity, short coherence times, and low gate fidelity necessitate optimizations tailored to specific platforms.

\subsection{Neutral-Atom Quantum Computing (NAQC)}
NAQC has emerged as a promising platform for scalable quantum computation. This technology uses neutral atoms, typically alkali atoms like rubidium or cesium, as qubits. The atoms are trapped and manipulated using lasers and magnetic fields, enabling precise control over quantum states. Single-qubit operations are performed using Raman laser pulses~\cite{beugnon2007two, graham_multi-qubit_2022, Levine_2022}, while multi-qubit gates are implemented through the Rydberg blockade mechanism~\cite{evered2023highfidelity,prl19-levine}. In this process, two atoms within the Rydberg interaction radius (\(r_{\mathrm{int}}\)) accumulate a controlled phase upon simultaneous excitation to a Rydberg state. This interaction facilitates high-fidelity two-qubit gates, such as the controlled-Z (\(CZ\)) gate, which is a common building block in quantum circuits.

A distinguishing feature of NAQC is its dynamic qubit reconfigurability, enabled by the physically mobile nature of neutral atoms. Using acousto-optic deflectors (AODs) and spatial light modulators (SLMs), atoms can be repositioned dynamically within a two-dimensional grid. This flexibility allows for arbitrary qubit interactions without requiring SWAP gates, setting NAQC apart from more rigid architectures like superconducting qubits. However, the need for atom movement introduces a unique optimization challenge: minimizing the time and distance of qubit transport to maintain coherence and fidelity.

The Dynamically Field-Programmable Qubit Array (DPQA) architecture exemplifies the potential of NAQC systems~\cite{Bluvstein_2022,Bluvstein_2023,reichardt2024}. On DPQA devices, qubits are held in two types of traps:
\begin{itemize}
    \item \textbf{Static traps}: Generated by SLMs, these traps provide stable positions for qubits, spaced (e.g., \(2.5 \times r_{\mathrm{int}}\)) to minimize crosstalk.
    \item \textbf{Mobile traps}: The AOD forms traps at the intersections of a grid of rows and columns, where each row or column coordinate can be independently activated, moved, and deactivated, enabling arbitrary rearrangements of atoms.
\end{itemize}

During quantum operations, mobile traps transport atoms to align with the requirements of the circuit, such as executing a \(CZ\) gate. See Fig.~\ref{fig:fpqa} for an illustration. To achieve high performance, these movements must be carefully coordinated to minimize total distance and avoid collisions.

\subsection{Compilation Challenges in NAQC} \label{sec:naqct}
Quantum circuit compilation for NAQC involves converting high-level quantum algorithms into hardware-executable instructions while addressing the platform's unique characteristics. In superconducting architectures, qubits are fixed, and connectivity limitations are managed through SWAP gates, which significantly increase circuit depth and reduce fidelity. NAQC eliminates the need for SWAP gates but shifts the focus to optimising atom movement. Key challenges include:
\begin{itemize}
    \item \textbf{Minimising movement time and distance}: Excessive or poorly optimized movements can lead to decoherence and gate errors.
    \item \textbf{Maintaining gate parallelism}: Leveraging the ability to execute gates simultaneously during global Rydberg illumination requires precise scheduling of movements.
\end{itemize}
Addressing these challenges is essential for realising the full potential of NAQC. By optimizing qubit placement and routing strategies, it is possible to achieve high-fidelity circuit execution while minimizing resource overheads.

\subsection{Compilation Procedures in NAQC}\label{sec:meet-interact-separate}
Consider a quantum circuit \(C\) consisting of single- and two-qubit gates that are native to a DPQA device $M$. In NAQC, single-qubit gates are implemented independently using Raman laser pulses and are therefore typically excluded from \(C\) when assessing the cost of atom movement. This simplification allows the problem to focus on the native two-qubit gates, which are typically \(CZ\) gates. These \(CZ\) gates are organized into layers, each of which can, in principle, be executed in parallel through global Rydberg illumination.

Initially, each logical qubit in \(C\) is mapped to a physical atom in the DPQA device $M$. For each layer of parallel \(CZ\) gates, any pair of interacting qubits must be co-located at the same  site (i.e., grid point) during the gate operation and separated afterward. As illustrated in Fig.~\ref{fig:atom-move}, the execution of any \(CZ\) gate (or several \(CZ\) gates in parallel)  involves the following steps:
\begin{enumerate}
    \item \textbf{Meet}: Move one or both atoms to the same site from their previous positions.
    \item \textbf{Interact}: Apply a global Rydberg laser to interact the qubits at the same position, executing the corresponding \(CZ\) gates in parallel.
    \item \textbf{Separate}: After the interaction, move the atoms apart so that they occupy different sites. 
    %This step may involve returning the atoms to their original positions, swapping their positions, or moving them to new positions.
\end{enumerate}
The execution of each \(CZ\) involves at least \emph{two} atom moves: a `meet' step and a ``separate'' step. These processes can often be parallelized across several or all gates within the same layer. The maximum distance moved during these steps determines the optimal distance required. In the DPQA device $M$, this maximal distance is at least $d$ in each step, where $d$ is the unit distance of the device.  Consequently, a circuit with $S$ layers of $CZ$ gates requires at least $2S$ atom moves, covering a total distance of at least $2S d$.

\begin{figure}[htbp]
    \centering
    \renewcommand{\arraystretch}{2} % Adjust row spacing

    % Set a total table width
    \begin{tabular}{l|>{\raggedright\arraybackslash}m{0.3\columnwidth}|>{\raggedright\arraybackslash}m{0.3\columnwidth}}
        & Restore mode & Swap mode \\
        \hline
        1. Initial & 
        \parbox[c]{\linewidth}{
            \begin{tikzpicture}
                \node[control atom] (initial_control) at (0,0) {};
                \node[target atom, right=2 of initial_control] (initial_target) {};
            \end{tikzpicture}
        } &
        \parbox[c]{\linewidth}{ 
            \begin{tikzpicture}
                \node[control atom] (initial_control) at (0,0) {};
                \node[target atom, right= 2 of initial_control] (initial_target) {};
            \end{tikzpicture}
        } \\
        \hline
        2. Meet & 
        \parbox[c]{\linewidth}{ 
            \begin{tikzpicture}
                \node[control atom] (meet_control) at (0,0) {};
                \node[target atom, right=1 of meet_control] (meet_target) {};
                \node[dashed node, right=2 of meet_control] {};
                \draw[->, dashed] (meet_target) -- ++(-1,0);
            \end{tikzpicture}
        } &
        \parbox[c]{\linewidth}{ 
            \begin{tikzpicture}
                \node[control atom] (meet_control) at (0,0) {};
                \node[target atom, right=1 of meet_control] (meet_target) {};
                \node[dashed node, right=2 of meet_control] {};
                \draw[->, dashed] (meet_target) -- ++(-1,0);
            \end{tikzpicture}
        } \\
        \hline
        3. Interact & 
        \parbox[c]{\linewidth}{ 
            \begin{tikzpicture}
                \node[control atom] (interact_control) at (0,0) {};
                \node[target atom, right=0.1 of interact_control] (interact_target) {};
            \end{tikzpicture}
        } &
        \parbox[c]{\linewidth}{ 
            \begin{tikzpicture}
                \node[control atom] (interact_control) at (0,0) {};
                \node[target atom, right=0.1 of interact_control] (interact_target) {};
            \end{tikzpicture}
        } \\
        \hline
        4. Separate & 
        \parbox[c]{\linewidth}{ 
            \begin{tikzpicture}
                \node[control atom] (separate_control) at (0,0) {};
                \node[dashed node, right=0.1 of separate_control] {};
                \node[target atom, right=1 of separate_control] (separate_target) {};
                \draw[->, dashed] (separate_target) -- ++(0.5,0);
            \end{tikzpicture}
        } &
        \parbox[c]{\linewidth}{ 
            \begin{tikzpicture}
                \node[dashed node] (origin_control) at (0,0) {};
                \node[target atom, right=0.1 of origin_control] (separate_target) {};
                \node[control atom, right=1 of origin_control] (separate_control) {};
                \draw[->, dashed] (separate_control) -- ++(0.5,0);
            \end{tikzpicture}
        } \\
        \hline
        5. Final & 
        \parbox[c]{\linewidth}{ 
            \begin{tikzpicture}
                \node[control atom] (final_control) at (0,0) {};
                \node[target atom, right=2 of final_control] (final_target) {};
            \end{tikzpicture}
        } &
        \parbox[c]{\linewidth}{ 
            \begin{tikzpicture}
                \node[target atom] (final_target) at (0,0) {};
                \node[control atom, right=2 of final_target] (final_control) {};
            \end{tikzpicture}
        } \\
    \end{tabular}

    \caption{The meet-interact-separate steps for executing a $CZ$ gate in a DPQA device. In \textit{restore mode}, atoms return to their original positions after interaction, while in \textit{swap mode}, their positions are exchanged. Offset movements may be required to prevent collisions when two atoms occupy the same site.
    %\cyan{I updated this figure by reducing the vetical space and modifying the left subfigure. I incorporated patterns to distinguish different nodes, ensuring that those using black-and-white printers can easily recognize the differences. If this adjustment is unnecessary, the style can remain consistent with the right figure.}
    }
    \label{fig:atom-move}
\end{figure}

We note that the ``separate'' step is highly flexible, a flexibility that plays a crucial role in this paper. Specifically, we distinguish between two modes of this step: the ``swap'' mode exchanges the positions of the two interacting atoms, and the ``restore'' mode, which returns the two interacting atoms to their original positions. See Figure~\ref{fig:atom-move} for illustrations of these modes. The moves in the meet and separate steps are called \emph{big moves} in this paper, which involve significant atom movement between different grid point. We also note, for interaction purpose, two atoms may locate at each grid point in the same time. To avoid atom collision, we need move one atom a little bit away from the grid point. This kind of offset movements are also required when we, for example, swap the positions of two atoms (see the right of Figure~\ref{fig:atom-move}).

\section{Overview of Quantum Fourier Transform}
\label{sec:qft}

QFT is a fundamental component of many important quantum algorithms. Efficiently implementing QFT on quantum hardware is crucial for realizing the potential of these algorithms. However, the all-to-all interaction pattern of QFT poses significant challenges, especially in architectures with limited qubit connectivity. In this section, we review the structure of QFT circuits and discuss the optimal transformation strategies for linear superconducting architectures, which will serve as a foundation for our proposed methods in NAQC.

\subsection{Structure of the Quantum Fourier Transform}
For an $n$-qubit system, the QFT circuit consists of $n(n-1)/2$ controlled-phase gates, interleaved with single-qubit Hadamard operations. These gates are arranged in a specific pattern that transforms the input state into its transformed state. Fig.~\ref{fig:qft-5} illustrates a standard QFT circuit with $n = 5$ qubits.

\begin{figure}[hbtp]
    \centering
    \begin{minipage}[b]{0.65\columnwidth}
        \centering
        \resizebox{!}{20ex}{
        \begin{tikzpicture}
            \begin{yquant}
                qubit {$q_{\idx}$} q[5];
                h q[0];
                box {$P$} q[1] | q[0];
                box {$P$} q[2] | q[0];
                box {$P$} q[3] | q[0];
                box {$P$} q[4] | q[0];
                h q[1];
                box {$P$} q[2] | q[1];
                box {$P$} q[3] | q[1];
                box {$P$} q[4] | q[1];
                h q[2];
                box {$P$} q[3] | q[2];
                box {$P$} q[4] | q[2];
                h q[3];
                box {$P$} q[4] | q[3];
                h q[4];
            \end{yquant}
        \end{tikzpicture}
        }
        \subcaption{}
        \label{fig:qft-5}
    \end{minipage}
    
    \begin{minipage}[b]{0.8\columnwidth}
        \centering
        \resizebox{!}{20ex}{
        \begin{tikzpicture}
            \begin{yquant}
                qubit {$q_{\idx}$} q[5];
                box {$P$} q[1] | q[0];
                [thick] barrier (q);
                box {$P$} q[2] | q[0];
                [thick] barrier (q);
                box {$P$} q[2] | q[1];
                box {$P$} q[3] | q[0];
                [thick] barrier (q);
                box {$P$} q[3] | q[1];
                box {$P$} q[4] | q[0];
                [thick] barrier (q);
                box {$P$} q[4] | q[1];
                box {$P$} q[3] | q[2];
                [thick] barrier (q);
                box {$P$} q[4] | q[2];
                [thick] barrier (q);
                box {$P$} q[4] | q[3];
                [thick] barrier (q);
            \end{yquant}
        \end{tikzpicture}}
        \subcaption{}
        \label{fig:qft-5-par}
    \end{minipage}
    
    \begin{minipage}[b]{\columnwidth}
        \centering
        \resizebox{!}{20ex}{\begin{tikzpicture}
\tikzset{
    Diam/.style={
        draw,
        diamond,
        inner sep=2pt,
        fill=black,
    },
    upper/.style={
        fill=black,
        inner sep=1.7pt,
        shape=dart,
        shape border rotate=90,
    },
    lower/.style={
        fill=black,
        inner sep=1.7pt,
        shape=dart,
        shape border rotate=270,
    }
}
\begin{yquant*}[operator/separation=0cm]
qubit {$(Q_{\idx})$} q[5];
hspace {0.1cm} q;
inspect {$q_{0}$} q[0];
inspect {$q_{1}$} q[1];
inspect {$q_{2}$} q[2];
inspect {$q_{3}$} q[3];
inspect {$q_{4}$} q[4];
%[thick] barrier (q);
box {$P$} q[1] | q[0];
% inspect {$q_{0}$} q[0];
% inspect {$q_{1}$} q[1];
swap (q[0],q[1]);
inspect {$q_{0}$} q[1];
inspect {$q_{1}$} q[0];
[thick] barrier (q);
box {$P$} q[2] | q[1];
% inspect {$q_{0}$} q[1];
% inspect {$q_{2}$} q[2];
swap (q[1],q[2]);
inspect {$q_{0}$} q[2];
inspect {$q_{2}$} q[1];
[thick] barrier (q);
box {$P$} q[1] | q[0];
box {$P$} q[3] | q[2];
% inspect {$q_{0}$} q[2];
% inspect {$q_{3}$} q[3];
% inspect {$q_{1}$} q[0];
% inspect {$q_{2}$} q[1];
swap (q[0],q[1]);
swap (q[2],q[3]);
inspect {$q_{0}$} q[3];
inspect {$q_{3}$} q[2];
inspect {$q_{1}$} q[1];
inspect {$q_{2}$} q[0];
[thick] barrier (q);
box {$P$} q[2] | q[1];
box {$P$} q[4] | q[3];
% inspect {$q_{0}$} q[3];
% inspect {$q_{4}$} q[4];
% inspect {$q_{1}$} q[1];
% inspect {$q_{3}$} q[2];
swap (q[1],q[2]);
swap (q[3],q[4]);
inspect {$q_{0}$} q[4];
inspect {$q_{4}$} q[3];
inspect {$q_{1}$} q[2];
inspect {$q_{3}$} q[1];
[thick] barrier (q);
box {$P$} q[1] | q[0];
box {$P$} q[3] | q[2];
% inspect {$q_{4}$} q[3];
% inspect {$q_{1}$} q[2];
% inspect {$q_{3}$} q[1];
% inspect {$q_{2}$} q[0];
swap (q[0],q[1]);
swap (q[2],q[3]);
inspect {$q_{4}$} q[2];
inspect {$q_{1}$} q[3];
inspect {$q_{3}$} q[0];
inspect {$q_{2}$} q[1];
[thick] barrier (q);
box {$P$} q[2] | q[1];
% inspect {$q_{4}$} q[2];
% inspect {$q_{2}$} q[1];
swap (q[1],q[2]);
inspect {$q_{4}$} q[1];
inspect {$q_{2}$} q[2];
[thick] barrier (q);
box {$P$} q[1] | q[0];
% inspect {$q_{3}$} q[0];
% inspect {$q_{4}$} q[1];
swap (q[0],q[1]);
% inspect {$q_{3}$} q[1];
% inspect {$q_{4}$} q[0];
% [thick] barrier (q);

null q[4];
null q[4];
inspect {$q_{0}$} q[4];
null q[3];
null q[3];
inspect {$q_{1}$} q[3];
null q[2];
null q[2];
inspect {$q_{2}$} q[2];
inspect {$q_{3}$} q[1];
inspect {$q_{4}$} q[0];
\end{yquant*}

% [red] inspect {[red] ?
% Add a legend right
% Define legend position
% \matrix[anchor=north] (legend) at (current bounding box.south)  {
%     % Upper filled diamond
%     \node[upper] (upper) {};
    
%     % Lower outlined diamond
%     \node[lower, below=6pt of upper] (lower) {};
    
%     % Connecting line
%     \draw (upper.south) -- (lower.north);
%     % First row: Diamond Gate (Black)
%     &
%     \node[align=center] {Two-qubit gate in swap mode};
%     &
    
%     &
%      % Upper filled diamond
%     \node[Diam] (upper) {};
    
%     % Lower outlined diamond
%     \node[Diam, below=6pt of upper] (lower) {};
    
%     % Connecting line
%     \draw (upper.south) -- (lower.north);
%     % First row: Diamond Gate (Black)
%     &
%     % \node[align=center] {Two-qubit \\ gate}; 
%     % &
%     \node[align=center] {Two-qubit gate in restore mode};
%     \\
% };
\end{tikzpicture}}
        \subcaption{}
        \label{fig:maslov}
    \end{minipage}
    
    \begin{minipage}[b]{.6\columnwidth}
    \centering
    \resizebox{\columnwidth}{!}{
         \begin{yquantgroup}
             \registers{
             qubit {} q[2];
             }
             \circuit{
             box {$P(\lambda)$} q[1] | q[0];
             }
             \equals
             \circuit{
             box {$Rz(\lambda/2)H$} q[1];
             zz (q[0], q[1]);
             box {$Rx(-\lambda/2)$} q[1];
             zz (q[0], q[1]);
             H q[1];
             box {$P(\lambda/2)$} q[0];
             }
        \end{yquantgroup}
    }
    \subcaption{}
    \label{fig:decom-cp}
    \end{minipage}
    \caption{(\subref{fig:qft-5}) Standard QFT-5 circuit, with each \(P\) denoting a controlled-phase gate. 
    (\subref{fig:qft-5-par}) Layered partition of QFT-5 (with single-qubit gates removed), showing 7 sequential controlled-phase layers.
    (\subref{fig:maslov}) The optimal transformation of QFT-5 on a linear superconducting architecture, with each mapping annotated to reflect the updated qubit arrangement after applying the necessary SWAP gates.
    (\subref{fig:decom-cp}) Decomposition of a controlled-phase gate into single-qubit gates and \(CZ\) gates.}
\end{figure}
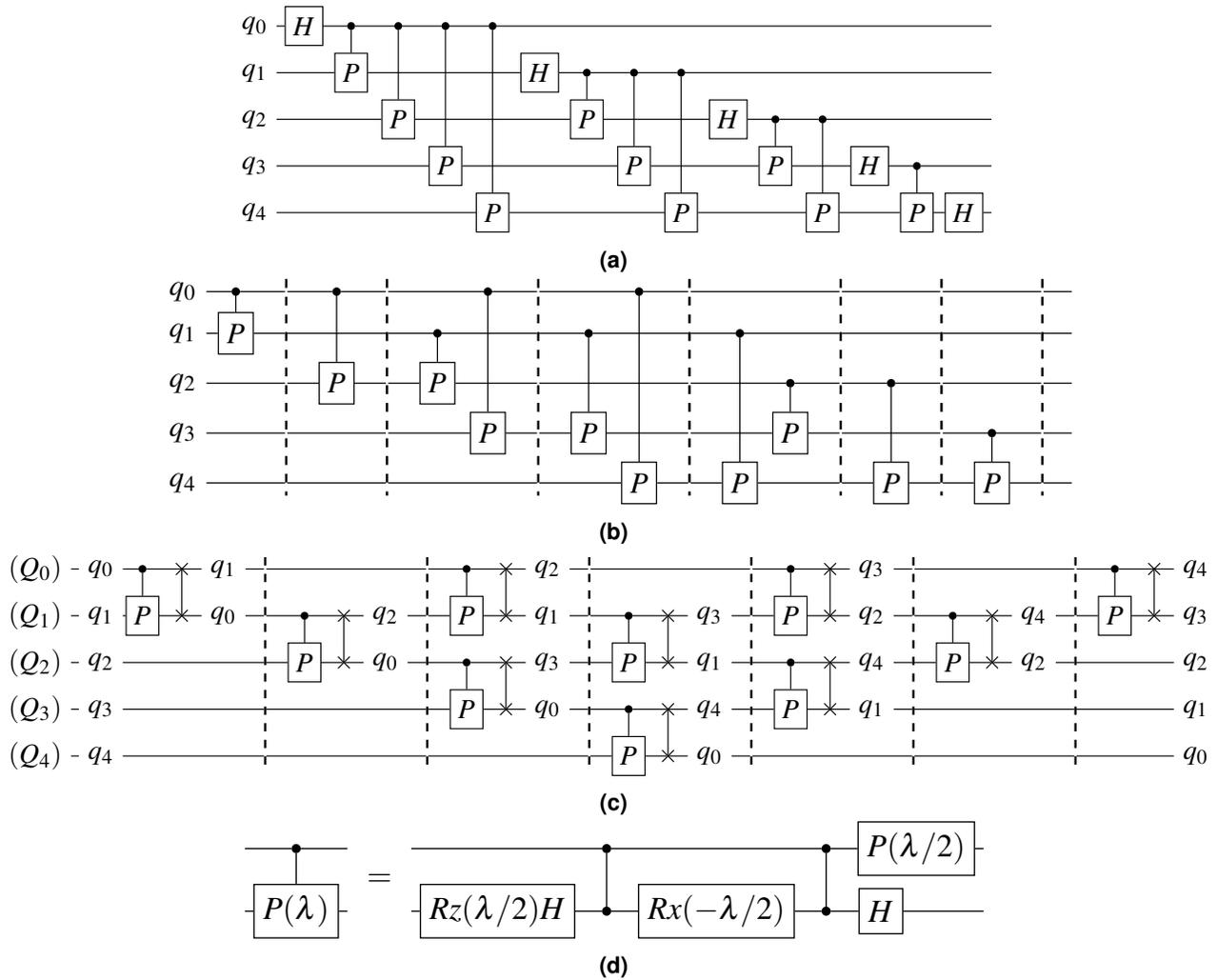

For circuit transformation purposes, we focus only on two-qubit gates in this paper. By pushing these gates as far forward as possible, the QFT circuit consists of $2n-3$ layers of two-qubit gates. Fig.~\ref{fig:qft-5-par} illustrates that QFT-5 has $2n-3=7$ two-qubit gate layers.

\subsection{Maslov's Optimal Transformation for Linear Superconducting Architectures}

The QFT-$n$ circuit, having a complete interaction graph, cannot be directly implemented on linear architectures or practical superconducting devices. %those shown in Fig.~\ref{fig:architectures}. 
In superconducting devices, this necessitates the use of SWAP gates to iteratively remap qubits until all two-qubit gates act on adjacent physical qubits.

For linear superconducting architectures, Maslov \cite{maslov+:physreva07} discovered an optimal transformation for the QFT-$n$ circuit (see Fig.~\ref{fig:maslov}). This transformation requires $2n-3$ layers of SWAP gates, resulting in a transformed circuit with linear depth in $n$. The transformation warrants close examination.

In Maslov's approach, we consider each layer of two-qubit gates, along with its accompanying SWAP gates, as a single \textbf{mapping stage} (m-stage for short). Each m-stage is associated with a mapping. Let $\tau_k$ denote the mapping at the $k$-th m-stage. Initially, $\tau_0$ maps logical qubit $q_i$ to physical qubit $Q_i$. We next provide a characterization for those SWAPs used in each m-stage.

For m-stages $0\leq k\leq n-2$, define
\begin{equation}\label{eq:Nk}
    N_k \equiv \set{j\in \{k, k-2, \ldots, k\!\! \pmod{2}}.
\end{equation} 
For m-stages $n-1\leq k \leq 2n-4$, define $N_k=N_{2n-4-k}$.

The following lemma characterizes the inserted SWAP gates at each m-stage.

\begin{lemma}
    At m-stage $k$, the two-qubit gates to be executed are $CP(r_j,s_j)$ for $j\in N_k$, where $r_j$ and $s_j$ are the logical qubits in the QFT-$n$ circuit that are mapped to physical qubits $Q_j$ and $Q_{j+1}$, respectively, under the mapping $\tau_{k}$. Thus, under the mapping $\tau_{k}$, these two-qubit gates act on physical qubits $Q_j$ and $Q_{j+1}$ for $j\in N_k$. Moreover, SWAP gates are inserted to exchange the states of physical qubits $Q_j$ and $Q_{j+1}$ for $j\in N_k$. 
\end{lemma}
For convenience, for each $0\leq k\leq 2n-4$, we write 
\begin{equation}\label{eq:Ek}
    E_k \equiv \set{(j,j+1) \mid j\in N_k}
\end{equation} 
and call edges in $E_k$ the \emph{relevant edges} in the $k$-th m-stage.

Through the SWAP operations corresponding to those relevant edges, the initial trivial mapping $\tau_0$ is transformed into the final mapping $\tau_{2n-3}$, which maps $q_0, \ldots,q_{n-1}$ to $Q_{n-1},\ldots,Q_0$. An illustration is given in Fig.~\ref{fig:maslov}. 

The following lemma demonstrates that the SWAP gates inserted in consecutive m-stages operate on alternating edges in the linear architecture.
\begin{lemma}\label{lem:alternating}
    For any odd $k$ and any even $k'$, $E_k\cap E_{k'}=\varnothing$. In particular, $E_k \cap E_{k+1}=\varnothing$ for any $0\leq k\leq 2n-5$. 
\end{lemma}
The lemma will play an important role in our design of optimal transformation for general DPQA architectures. 

While Maslov's transformation is optimal for linear superconducting architectures, it relies heavily on SWAP gates to facilitate interactions between non-adjacent qubits. In NAQC, however, SWAP gates can be replaced with atom movements, which offer a more flexible and potentially more efficient alternative. This key difference underscores the need for new compilation strategies tailored to the unique capabilities of NAQC platforms.

In the next section, we explore how these capabilities can be leveraged to develop efficient compilation strategies for NAQC platforms.

\section{Optimal Transformation of QFT Circuits in DPQA}\label{sec:opt}

In this section, we present optimal compilation strategies for implementing QFT circuits on DPQA architectures. We begin by estimating the theoretical lower bound on atom movements required for QFT circuits. Next, we introduce an optimal transformation for linear DPQA architectures, which serves as the foundation for our approach. Finally, we extend this transformation to 2D grid-like architectures, demonstrating how to maintain efficiency while adapting to practical hardware constraints. By leveraging the unique capabilities of DPQA, such as dynamic qubit reconfigurability, our methods achieve theoretical lower bounds in movement counts while preserving high circuit fidelity.

\subsection{Replacing SWAP Gates with Atom Movement}

In superconducting architectures, SWAP gates are commonly used to enable interactions between non-adjacent qubits. However, on DPQA architectures, SWAP gates can be entirely replaced with atom movements, which are more efficient and do not require additional gate operations. This is made possible by the dynamic reconfigurability of qubits in DPQA, where atoms can be physically moved to facilitate interactions.

Since the \(CZ\) gate is the native two-qubit gate on DPQA architectures, we first decompose the controlled-phase gates in the QFT circuit into \(CZ\) gates. As shown in Fig.~\ref{fig:decom-cp}, each controlled-phase gate is decomposed into two \(CZ\) gates. Consequently, the QFT-$n$ circuit, which originally consists of $2n-3$ layers of controlled-phase gates, is transformed into $2(2n-3)$ layers of \(CZ\) gates.

Recall from Section~\ref{sec:meet-interact-separate} %\cyan{\emph{Compilation Procedures in NAQC}} 
that the execution of a \(CZ\) gate on DPQA involves three steps: 1) moving atoms to the same grid point (meet step), 2) applying a global Rydberg laser to perform the gate (interact step), and 3) separating the atoms (separate step). Each execution of \(CZ\) gates thus requires at least two atom movements: one for the meet step and one for the separate step. Importantly, the ``separate'' step does not require the interacted atoms to be returned to their original positions. Instead, their positions can be exchanged, as illustrated in Fig.~\ref{fig:atom-move}. Shortly, we will see that this flexibility allows us to simulate SWAP operations without additional overhead, further optimizing the circuit.

For the QFT-$n$ circuit, which has $2(2n-3)$ layers of \(CZ\) gates, the \textbf{theoretical lower bound} on the total number of atom moves is $4(2n-3)$. Atom movements can potentially be fully parallelized across gates within the same layer, minimizing the total distance traveled by the atoms. Achieving this bound requires careful scheduling of atom movements to ensure that all gates in a layer are executed with minimal overhead.

\subsection{Optimal Transformation in a Linear DPQA Architecture}\label{sec:linear}

In this subsection, we propose an optimal transformation for QFT circuits on a linear DPQA architecture, which can be viewed as a $1 \times n$ grid. Each grid point is represented by its coordinate $P_i \equiv (i, 0)$, and the unit distance between consecutive grid points is $d$.

The transformation builds on Maslov's approach for linear superconducting architectures, where SWAP gates are inserted to enable interactions between non-adjacent qubits (see Fig.~\ref{fig:maslov}). In DPQA, we replace these SWAP gates with atom movements, which are more efficient. Specifically, for each m-stage (mapping stage) in Maslov's transformation, the layer of controlled-phase gates is replaced with two identical layers of $CZ$ gates. We then perform the following steps for each m-stage:
\begin{enumerate}[label=\alph*.]
    \item \textbf{First \(\mathit{CZ}\)-layer}:
        \begin{itemize}
            \item Move the atoms involved in the gates to the same grid point (meet step).
            \item Apply the gates using a global Rydberg laser (interact step).
            \item Move the atoms back to their original positions (separate step in restore mode).
        \end{itemize}
    \item \textbf{Second \(\mathit{CZ}\)-layer}:
        \begin{itemize}
            \item Move the atoms involved in the gates to the same grid point (meet step).
            \item Apply the gates using a global Rydberg laser (interact step).
            \item Swap the positions of the atoms (separate step in swap mode).
        \end{itemize}
\end{enumerate}

This approach ensures that all \(CZ\) gates are executed with minimal atom movements, achieving the theoretical lower bound of $4(2n-3)d$. The transformation is optimal because it parallelizes movements within each m-stage and minimizes the total distance traveled by the atoms.

More precisely, for edges \(e_j = (j, j+1)\) with \(j \in E_k\) (as defined in \eqref{eq:Ek}), the atom movements are performed as follows: in Step a, we move all smaller-indexed atoms towards the larger-indexed atoms and move them back; in Step b, we move all smaller-indexed atoms towards the large-indexed atoms and swap their positions. 

The movement directions are consistent within each batch of atom movements, and each move covers only a unit distance \(d\). This ensures that the total maximal distance traveled by the atoms is \(4(2n-3)d\), which matches the theoretical lower bound. Therefore, the transformation on the linear DPQA architecture is optimal.

As a consequence, we have the following result:
\begin{theorem}
    Let \(M\) be a linear DPQA architecture of \(n\) qubits. The transformation described above for QFT-\(n\) is an optimal transformation on \(M\).
\end{theorem}

To illustrate the transformation, we revisit the QFT-5 circuit. 

\begin{example}
The QFT-5 circuit consists of 7 m-stages and 14 layers of \(CZ\) gates. Initially, the logical qubits \(q_i\) are mapped to the atom located at the grid point \(P_i = (i, 0)\) for \(0 \leq i \leq 4\) on the linear architecture. Fig.~\ref{fig:qft5-lnn} illustrates the mapping stages and movement sequences for QFT-5. 

For example, at m-stage 3, the set of edges \(E_3 = \set{(1,2), (3,4)}\) indicates that atom movements are performed along these edges. Specifically:
\begin{itemize}
    \item In the first layer of \(CZ\) gates (Step a), the atoms at \(P_1\) and \(P_3\) are moved (along the same direction), respectively, to \(P_2\) and \(P_4\), the \(CZ\) gates are applied, and the atoms are returned to their original positions. 
    \item In the second layer of \(CZ\) gates (Step b), the atoms at \(P_1\) and \(P_3\) are moved (along the same direction), respectively, to \(P_2\) and \(P_4\), the \(CZ\) gates are applied, and their positions are swapped. 
\end{itemize}

The red edges in Fig.~\ref{fig:qft5-lnn} highlight the atom movements along the edges \(e \in E_k\) for each m-stage \(k\) (where \(0 \leq k \leq 6\)), demonstrating how the transformation achieves the theoretical lower bound on movement counts.
\end{example}

\begin{figure}[hbtp]
    \centering
    \resizebox{!}{30ex}{\newcommand{\drawcircuitQFTF}[5]{
    % Custom layout coordinates
    \coordinate (q0) at #1;
    \coordinate (q1) at #2;
    \coordinate (q2) at #3;
    \coordinate (q3) at #4;
    \coordinate (q4) at #5;

    % Draw the points
    \foreach \i in {0,1,...,4} {
        \node[point] (q\i) at (q\i) {}; 
        \node[pointLabel] at (q\i) {q\i};
    }
}
\begin{tikzpicture}[scale=1,
    point/.style={
        fill=blue!30,       % Color of the points
        draw=black,       % Outline color (if needed)
        shape=circle,     % Shape of the point
        minimum size=10pt, % Adjust the point size
        inner sep=0pt
    },
    pointLabel/.style={
        above left,        % Position of the label relative to the point
        font=\large        % Font size for the labels
    }]
    % Axes
    \draw[thick,->] (-1.2,-0.4) -- (-1.2, 4.9) node[left] {y};
    \draw[thick,->] (-1.2,-0.4) -- (7.9, -0.4) node[right] {x};
    % axis ticks
    \foreach \x in {0,1.2,2.4,3.6,4.8,6,7.2} {
       \draw (\x,-0.4) -- (\x,-0.3);
    }
    \node[below] at (0,-0.4) {0};
    \node[below] at (1.2,-0.4) { 1};
    \node[below] at (2.4,-0.4) { 2};
    \node[below] at (3.6,-0.4) {3};
    \node[below] at (4.8,-0.4) {4};
    \node[below] at (6,-0.4) {5};
    \node[below] at (7.2,-0.4) {6};
    % y-axis
    \begin{scope}[shift={(-1.2,0)}]
        \foreach \i in {0,1,...,4} {
            \node[point,fill=blue!10,dashed] at (0,\i) {};
        }
    \end{scope}
    
    % Step 0
    \begin{scope}[shift={(0,0)}]
        \drawcircuitQFTF{(0,4)}{(0,3)}{(0,2)}{(0,1)}{(0,0)}
        \draw[red, thick] (q0) -- (q1);
    \end{scope}
    
    % Step 1
    \begin{scope}[shift={(1.2,0)}]
        \drawcircuitQFTF{(0,3)}{(0,4)}{(0,2)}{(0,1)}{(0,0)}
        \draw[red, thick] (q0) -- (q2);
    \end{scope}
    
    % Step 2
    \begin{scope}[shift={(2.4,0)}]
        \drawcircuitQFTF{(0,2)}{(0,4)}{(0,3)}{(0,1)}{(0,0)}
        \draw[red, thick] (q0) -- (q3);
        \draw[red, thick] (q1) -- (q2);
    \end{scope}

    % Step 3
    \begin{scope}[shift={(3.6,0)}]
        \drawcircuitQFTF{(0,1)}{(0,3)}{(0,4)}{(0,2)}{(0,0)}
        \draw[red, thick] (q0) -- (q4);
        \draw[red, thick] (q1) -- (q3);
    \end{scope}
    
    % Step 4
    \begin{scope}[shift={(4.8,0)}]
        \drawcircuitQFTF{(0,0)}{(0,2)}{(0,4)}{(0,3)}{(0,1)}
        \draw[red, thick] (q1) -- (q4);
        \draw[red, thick] (q2) -- (q3);
    \end{scope}

    % Step 5
    \begin{scope}[shift={(6,0)}]
        \drawcircuitQFTF{(0,0)}{(0,1)}{(0,3)}{(0,4)}{(0,2)}
        \draw[red, thick] (q2) -- (q4);
    \end{scope}

    % Step 6
    \begin{scope}[shift={(7.2,0)}]
        \drawcircuitQFTF{(0,0)}{(0,1)}{(0,2)}{(0,4)}{(0,3)}
        \draw[red, thick] (q3) -- (q4);
    \end{scope}
\end{tikzpicture}}
    \caption{Mapping stages (m-stages) and movement sequences for QFT-5 on the linear architecture, where the $x$-axis represents the mapping stage and the $y$-axis indicates the atom locations in the linear architecture. The red edges at each m-stage \(x\) (for \(0 \le x \le 6\)) highlight atom movements along edges \(e \in E_x\). 
    }
    \label{fig:qft5-lnn}
\end{figure}
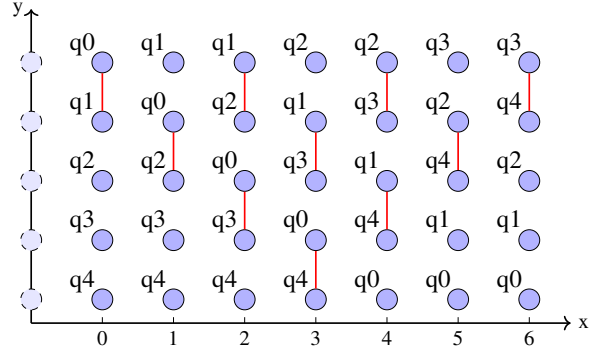

\subsection{Optimal Transformation on a Grid Architecture}
\label{sec:grid-transform}

While the linear architecture achieves minimal movement cost, it requires qubits to be arranged in a one-dimensional array, which becomes impractical for large numbers of qubits. To address this, we propose folding the linear architecture into a two-dimensional grid while preserving the efficiency of the linear transformation. This is achieved through a technique called \emph{zigzag folding}, which maps the linear arrangement of qubits onto a grid in a way that maintains adjacency and enables parallel gate execution.

\begin{definition}[Zigzag Folding]
\label{def-zigzag}
Let $P$ be a linear DPQA architecture of $n$ qubits $q_0, q_1, \ldots, q_{n-1}$, and let $M$ be a DPQA architecture with grid dimensions $m_1 \times m_2$. A \emph{zigzag folding} $\Phi$ of $P$ into $M$ is an injective mapping that assigns a unique grid point $(x_i, y_i)$ to each qubit $q_i$ such that:
\begin{itemize}
    \item \textbf{Neighboring Placement:} $\Phi(q_i)$ is a neighbor of $\Phi(q_{i-1})$ in $M$, meaning $|x_i - x_{i-1}| + |y_i - y_{i-1}| = 1$.
    \item \textbf{Alternating Direction:} If $q_i$ is placed horizontally relative to $q_{i-1}$, then $q_{i+1}$ (if $i < n-1$) is placed vertically relative to $q_i$, and vice versa.
\end{itemize}
\end{definition}

A zigzag folding ensures that neighboring qubits in the linear arrangement remain adjacent on the grid, enabling efficient gate execution. By Lemma~\ref{lem:alternating}, the gate operations and atom movements in each m-stage correspond to either horizontal or vertical edges in the zigzag folding, allowing for parallel execution.

\begin{theorem}\label{thm:zigzag_trans}
    Let $\Phi$ be a zigzag folding of a linear DPQA architecture of $n$ qubits into a grid DPQA architecture $M$. Then QFT-$n$ has an optimal transformation on $M$ as follows:
    \begin{enumerate}
        \item Initially, map each $q_i$ to $\Phi(q_i)$.
        \item For each m-stage $k$, let $E_k$ be the set of relevant edges as specified in \eqref{eq:Ek}. If the edges in $E_k$ are horizontal (vertical, resp.):
            \begin{itemize}
                \item For the first layer of $CZ$ gates in the m-stage, move all left (lower, resp.) atoms rightward (upward, resp.), execute the gates using a global Rydberg laser, and then separate the atoms by moving them back.
                \item For the second layer of $CZ$ gates in the m-stage, move all left (lower, resp.) atoms rightward (upward, resp.), execute the gates, and then separate the atoms by exchanging their positions.
            \end{itemize}
    \end{enumerate}
\end{theorem}

\begin{proof}
Without loss of generality, assume that \(q_1\) is placed as the right neighbor of \(q_0\) in the zigzag folding \(\Phi\). By the definition of zigzag folding, the placement of qubits alternates between horizontal and vertical directions. Consequently:
\begin{itemize}
    \item For even \(k\), all edges in \(E_k\) are horizontal.
    \item For odd \(k\), all edges in \(E_k\) are vertical.
\end{itemize}
For each m-stage \(k\), the transformation involves four batches of atom movements:
\begin{itemize}
    \item Two batches for the first layer of \(CZ\) gates (Step a).
    %    \begin{itemize}
    %        \item Move the atoms involved in the gates to the same grid point (meet step).
    %        \item Apply the \(CZ\) gates using a global Rydberg laser (interact step).
    %        \item Move the atoms back to their original positions (separate step in restore mode).
    %    \end{itemize}
    \item Two batches for the second layer of \(CZ\) gates (Step b).
    %    \begin{itemize}
    %        \item Move the atoms involved in the gates to the same grid point (meet step).
    %        \item Apply the \(CZ\) gates using a global Rydberg laser (interact step).
    %        \item Swap the positions of the atoms (separate step in swap mode).
    %    \end{itemize}
\end{itemize}
Since all movements within a batch occur in the same direction (either horizontal or vertical), they can be performed in parallel. Furthermore, each movement covers a unit distance \(d\), ensuring that the total distance traveled by the atoms is minimized. 

The maximal movement distance for the entire transformation is \(4(2n-3)d\), which matches the theoretical lower bound. This is achieved because:
\begin{itemize}
    \item Each of the \(2(2n-3)\) layers of \(CZ\) gates requires two atom movements (one for the meet step and one for the separate step).
    \item Each movement covers a unit distance \(d\).
    \item Movements are parallelized within each m-stage, ensuring no unnecessary overhead.
\end{itemize}

Neglecting offset movements, the transformation is optimal in terms of both the big move count and the total distance traveled.
\end{proof}

This transformation achieves the same theoretical lower bound on movement counts as the linear architecture while adapting to the practical constraints of 2D grid architectures.

Consider, for example, the two zigzag embeddings shown in Fig.~\ref{fig:twozigzag} for a linear architecture of five qubits. At m-stage 3, the set of relevant edges \(E_3 = \set{(1,2), (3,4)}\) defines the required interactions between qubits. In the left zigzag folding, both edges are horizontal, while in the right zigzag folding, both are vertical. This alignment ensures that the atom movements in m-stage 3 can be performed in parallel in both foldings. This parallelization of atom movements within each m-stage ensures that the transformation achieves the theoretical lower bound on movement counts while maintaining high circuit fidelity.

\subsection{Space-Efficient Zigzag Foldings}

While the zigzag folding approach guarantees optimal transformations for QFT circuits on grid-like architectures, it is also important to ensure that the folding is space-efficient. Specifically, we aim to implement QFT-$n$ compactly on a $(w \times w)$-grid with $w = O(\sqrt{n})$. In this subsection, we demonstrate how to construct such space-efficient zigzag foldings.

\subsubsection{Construction of Space-Efficient Zigzag Foldings}

\begin{figure}[htbp]
    \centering
    \raisebox{30pt}
    {\begin{minipage}[b]{0.3\columnwidth}
        \begin{minipage}{\columnwidth}
            \centering
            % \resizebox{.8\columnwidth}{!}{\input{figure/twozigzag}}
            \raisebox{8pt}{\begin{minipage}[b]{.45\columnwidth}
            \centering
            \resizebox{\columnwidth}{!}{\begin{tikzpicture}[scale=0.8,
    point/.style={
        fill=blue!30,       % Color of the points
        draw=black,       % Outline color (if needed)
        shape=circle,     % Shape of the point
        minimum size=15pt, % Adjust the point size
        inner sep=0pt
    }]

% First zigzag folding (left)
\begin{scope}[shift={(0,0)}]
    % Draw grid
    \draw[step=1, gray!50, very thin, dashed] (0,0) grid (2,2);
    
    % Draw zigzag path
    \draw[red] (0,2) -- (0,1);
    \draw[very thick, red] (0,1) -- (1,1);
    \draw[red] (1,1) -- (1,2);
    \draw[very thick, red] (1,2) -- (2,2);
    
    % Label points
    \node[point] at (0,2)  {\tiny $Q_0$};
    \node[point] at (0,1) {\tiny $Q_1$};
    \node[point] at (1,1) {\tiny $Q_2$};
    \node[point] at (1,2) {\tiny $Q_3$};
    \node[point] at (2,2) {\tiny $Q_4$};
\end{scope}

% Second zigzag folding (right)

\end{tikzpicture}}
            % \subcaption{}
            \label{fig:zz1}
            \end{minipage}}
            \begin{minipage}[b]{.45\columnwidth}
            \centering
            \resizebox{\columnwidth}{!}{\begin{tikzpicture}[scale=0.8,
    point/.style={
        fill=blue!30,       % Color of the points
        draw=black,       % Outline color (if needed)
        shape=circle,     % Shape of the point
        minimum size=15pt, % Adjust the point size
        inner sep=0pt
    }]

% Second zigzag folding (right)
\begin{scope}[shift={(0,0)}]
    % Draw grid
    \draw[step=1, gray!50, very thin, dashed] (0,0) grid (2,2);
    
    % Draw zigzag path
    \draw[red] (0,2) -- (1,2);
    \draw[thick,double, red] (1,2) -- (1,1);
    \draw[red] (1,1) -- (2,1);
    \draw[thick, double, red] (2,1) -- (2,0);
    
    % Label points
    \node[point] at (0,2)  {\tiny $Q_0$};
    \node[point] at (1,2) {\tiny $Q_1$};
    \node[point] at (1,1) {\tiny $Q_2$};
    \node[point] at (2,1) {\tiny $Q_3$};
    \node[point] at (2,0) {\tiny $Q_4$};
\end{scope}

\end{tikzpicture}}
            % \subcaption{}
            \label{fig:zz2}
            \end{minipage}
            \vspace{-20pt}
            \subcaption{}
            \label{fig:twozigzag}
        \end{minipage}
        \vspace{30pt}
        
        \begin{minipage}{\columnwidth}
            \centering
            \begin{minipage}[b]{0.1\columnwidth}
            \centering
            \resizebox{\columnwidth}{!}{
            \begin{tikzpicture}
              % Draw the black lines (the path)
              \draw (0,1) -- (0,0);
              \draw (0,1) -- (1,1);
              \draw (1,0) -- (1,1);
            \end{tikzpicture}
            }
            \subcaption{}
            % \subcaption{Illustrations of $\Phi_{1}$}
            \label{fig:Phi1}
        \end{minipage}
            \hspace{0.1\columnwidth}
            \begin{minipage}[b]{0.3\columnwidth}
            \centering
            \resizebox{\columnwidth}{!}{
            \begin{tikzpicture}
              % Draw the black lines (the path)
              \foreach \i in {1,2,3} {
                \draw (\i,4) -- (\i,3);
                \draw (\i,2) -- (\i,1);
                }
              \draw (4,2) -- (4,3);
              \draw (2,4) -- (3,4);
              \draw (2,1) -- (3,1);
              \foreach \j in {2,3} {
                \draw (1,\j) -- (2,\j);
                \draw (3,\j) -- (4,\j);
                }
            
              \fill[gray!40,draw=black,densely dotted] (4,1) circle (3pt);
              \fill[gray!40,draw=black,densely dotted] (4,4) circle (3pt);
            \end{tikzpicture}
            }
            \subcaption{}
            \label{fig:Phi2}
        \end{minipage}            
        \end{minipage}
    \end{minipage}}
% \caption{Illustrations of $\Phi_{1}$ (left) and $\Phi_{2}$ (right).}
    \begin{minipage}[b]{0.3\columnwidth}
    \centering
    \begin{tikzpicture}[scale=0.5]
    % Add a highlight box
      \fill[red!20,draw=red] (1.8,1.8) rectangle (7.2,7.2);
      \fill[blue!20,draw=blue] (3.8,3.8) rectangle (5.2,5.2);
      % Draw the black lines (the path)
      \draw (5,4) -- (5,5) -- (4,5) -- (4,4) -- (3,4) -- (3,5) -- (2,5) -- (2,6) -- (3,6) -- (3,7) -- (4,7) -- (4,6) -- (5,6) -- (5,7) -- (6,7) -- (6,6) -- (7,6) -- (7,5) -- (6,5) -- (6,4) -- (7,4) -- (7,3) -- (6,3) -- (6,2) -- (5,2) -- (5,3) -- (4,3) -- (4,2) -- (3,2) -- (3,3) -- (2,3) -- (2,2) -- (1,2) -- (1,3) -- (0,3) -- (0,4) -- (1,4) -- (1,5) -- (0,5) -- (0,6) -- (1,6) -- (1,7) -- (0,7) -- (0,8) -- (1,8) -- (1,9) -- (2,9) -- (2,8) -- (3,8) -- (3,9) -- (4,9) -- (4,8) -- (5,8) -- (5,9) -- (6,9) -- (6,8) -- (7,8) -- (7,9) -- (8,9) -- (8,8) -- (9,8) -- (9,7) -- (8,7) -- (8,6) -- (9,6) -- (9,5) -- (8,5) -- (8,4) -- (9,4) -- (9,3) -- (8,3) -- (8,2) -- (9,2) -- (9,1) -- (8,1) -- (8,0) -- (7,0) -- (7,1) -- (6,1) -- (6,0) -- (5,0) -- (5,1) -- (4,1) -- (4,0) -- (3,0) -- (3,1) -- (2,1) -- (2,0) -- (1,0) -- (1,1) -- (0,1) -- (0,0);
      \fill[gray!40,draw=black,densely dotted] (9,0) circle (4pt);
      \fill[gray!40,draw=black,densely dotted] (9,9) circle (4pt);
      \fill[gray!40,draw=black,densely dotted] (0,9) circle (4pt);
      \fill[gray!40,draw=black,densely dotted] (0,2) circle (4pt);
      \fill[gray!40,draw=black,densely dotted] (7,2) circle (4pt);
      \fill[gray!40,draw=black,densely dotted] (7,7) circle (4pt);
      \fill[gray!40,draw=black,densely dotted] (2,7) circle (4pt);
      \fill[gray!40,draw=black,densely dotted] (2,4) circle (4pt);
    \end{tikzpicture}
    \subcaption{}
    \label{fig:Phi5}
    \end{minipage}
    \begin{minipage}[b]{0.35\columnwidth}
        \centering
    \begin{tikzpicture}[scale=0.5]
    % Add a highlight box
      \fill[red!20,draw=red] (1.8,1.8) rectangle (9.2,9.2);
      \fill[blue!20,draw=blue] (3.8,3.8) rectangle (7.2,7.2);
      % Draw the black lines (the path)
      \draw (4,7) -- (4,6) -- (5,6) -- (5,7) -- (6,7) -- (6,6) -- (7,6) -- (7,5) -- (6,5) -- (6,4) -- (5,4) -- (5,5) -- (4,5) -- (4,4) -- (3,4) -- (3,5) -- (2,5) -- (2,6) -- (3,6) -- (3,7) -- (2,7) -- (2,8) -- (3,8) -- (3,9) -- (4,9) -- (4,8) -- (5,8) -- (5,9) -- (6,9) -- (6,8) -- (7,8) -- (7,9) -- (8,9) -- (8,8) -- (9,8) -- (9,7) -- (8,7) -- (8,6) -- (9,6) -- (9,5) -- (8,5) -- (8,4) -- (9,4) -- (9,3) -- (8,3) -- (8,2) -- (7,2) -- (7,3) -- (6,3) -- (6,2) -- (5,2) -- (5,3) -- (4,3) -- (4,2) -- (3,2) -- (3,3) -- (2,3) -- (2,2) -- (1,2) -- (1,3) -- (0,3) -- (0,4) -- (1,4) -- (1,5) -- (0,5) -- (0,6) -- (1,6) -- (1,7) -- (0,7) -- (0,8) -- (1,8) -- (1,9) -- (0,9) -- (0,10) -- (1,10) -- (1,11) -- (2,11) -- (2,10) -- (3,10) -- (3,11) -- (4,11) -- (4,10) -- (5,10) -- (5,11) -- (6,11) -- (6,10) -- (7,10) -- (7,11) -- (8,11) -- (8,10) -- (9,10) -- (9,11) -- (10,11) -- (10,10) -- (11,10) -- (11,9) -- (10,9) -- (10,8) -- (11,8) -- (11,7) -- (10,7) -- (10,6) -- (11,6) -- (11,5) -- (10,5) -- (10,4) -- (11,4) -- (11,3) -- (10,3) -- (10,2) -- (11,2) -- (11,1) -- (10,1) -- (10,0) -- (9,0) -- (9,1) -- (8,1) -- (8,0) -- (7,0) -- (7,1) -- (6,1) -- (6,0) -- (5,0) -- (5,1) -- (4,1) -- (4,0) -- (3,0) -- (3,1) -- (2,1) -- (2,0) -- (1,0) -- (1,1) -- (0,1) -- (0,0);
      \fill[gray!40,draw=black,densely dotted] (11,0) circle (4pt);
      \fill[gray!40,draw=black,densely dotted] (11,11) circle (4pt);
      \fill[gray!40,draw=black,densely dotted] (0,11) circle (4pt);
      \fill[gray!40,draw=black,densely dotted] (0,2) circle (4pt);
      \fill[gray!40,draw=black,densely dotted] (9,2) circle (4pt);
      \fill[gray!40,draw=black,densely dotted] (9,9) circle (4pt);
      \fill[gray!40,draw=black,densely dotted] (2,9) circle (4pt);
      \fill[gray!40,draw=black,densely dotted] (2,4) circle (4pt);
      \fill[gray!40,draw=black,densely dotted] (7,4) circle (4pt);
      \fill[gray!40,draw=black,densely dotted] (7,7) circle (4pt);
    \end{tikzpicture}
    \subcaption{}
    \label{fig:Phi6}
    \end{minipage}
    \caption{
    (\subref{fig:twozigzag}) Two zigzag foldings of the same five-qubit linear architecture into a \(3 \times 3\) grid. 
    At m-stage 3, the relevant edges \(\{(1,2),(3,4)\}\) run horizontally in the left folding and vertically in the right one. 
    In both cases, movements along these edges can occur in parallel, preserving the optimal total big moves.
    (\subref{fig:Phi1}) illustrates $\Phi_{1}$, while (\subref{fig:Phi2}) depicts $\Phi_{2}$. (\subref{fig:Phi5}) presents the zigzag path $\Phi_{5}$ within a $10 \times 10$ grid, where the red and blue boxes highlight the configurations of $\Phi_{3}$ and $\Phi_{1}$, respectively. Similarly, (\subref{fig:Phi6}) showcases the zigzag path $\Phi_{6}$ in a $12 \times 12$ grid, with the red and blue boxes highlighting the configurations of $\Phi_{4}$ and $\Phi_{2}$, respectively. The unvisited (or ``wasted'') grid points are highlighted.}
\label{fig:Phi}
\end{figure}
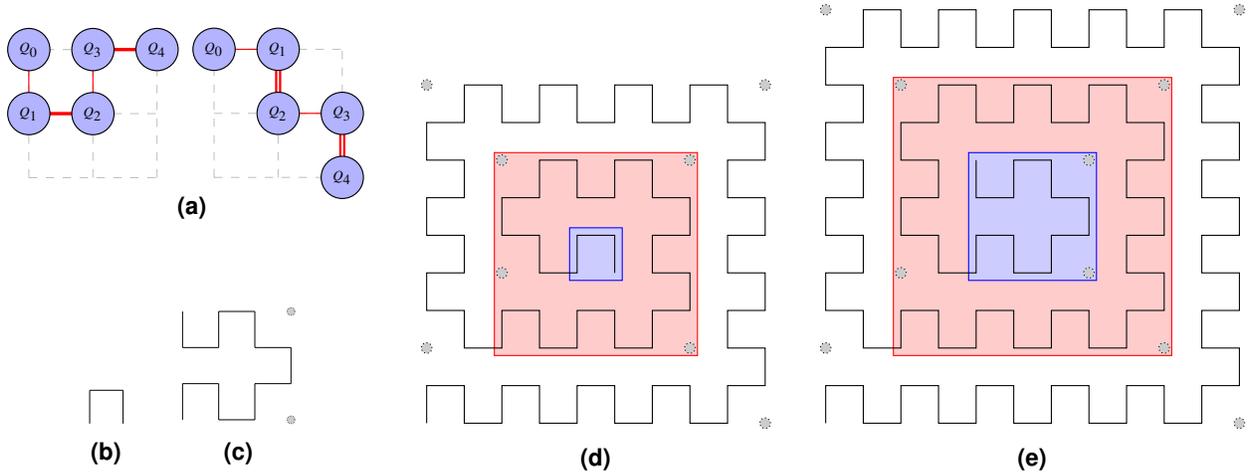

For each integer $w$, we construct a zigzag folding $\Phi_w$ in a $(2w \times 2w)$-grid. The construction proceeds as follows:
\begin{itemize}
    \item For $w = 1$, the zigzag folding $\Phi_1$ occupies 4 grid points in a $2 \times 2$ grid. The path starts at $P_0 = (0, 0)$, moves up to $P_1 = (0, 1)$, then right to $P_2 = (1, 1)$, and finally down to $P_3 = (1, 0)$. See Fig.~\ref{fig:Phi2} for an illustration.
    \item For $w = 2$, the zigzag folding $\Phi_2$ occupies 12 grid points in a $4 \times 4$ grid. See Fig.~\ref{fig:Phi2} for an illustration, where only two grid points are wasted.
    \item For $w \geq 3$, the zigzag folding $\Phi_w$ is constructed by extending $\Phi_{w-2}$ from a $(2w-4) \times (2w-4)$ grid to a $(2w \times 2w)$ grid. The path alternates between horizontal and vertical movements, ensuring that neighboring qubits in the linear arrangement remain adjacent on the grid. See Fig.~\ref{fig:Phi5} and Fig.~\ref{fig:Phi6} for illustrations of $\Phi_5$ and $\Phi_6$, respectively.
\end{itemize}
We note that the location does not matter and a configuration can be translated or properly rotated in the plane.

\subsubsection{Space Efficiency Analysis}

Let $\phi(m)$ denote the number of empty grid points (i.e., grid points not visited by the path $\Phi_m$) in a $(2m \times 2m)$ grid. We observe that:
\begin{itemize}
    \item $\phi(1) = 0$ (all grid points are occupied).
    \item $\phi(2) = 2$ (two grid points are empty).
    \item For $m \geq 1$, $\phi(m+2) = \phi(m) + 4$.
\end{itemize}

By induction, it follows that $\phi(m) \leq 2m$. This implies that the number of wasted grid points grows linearly with $m$, while the total number of grid points grows quadratically. As a result, the space efficiency of the zigzag folding improves as $n$ increases.

\begin{lemma}
    For any $n$, we can construct a zigzag folding path of $n$ atoms in a $(2m+2) \times (2m+2)$ grid such that the spatial efficiency approaches 1 as $n$ increases, where $m =O(\sqrt{n})$.
\end{lemma}

\begin{proof}
    Let $m = \lceil \sqrt{n}/2 \rceil$. The zigzag folding $\Phi_{m+1}$ in a $(2m+2) \times (2m+2)$ grid has at most $2m+2$ wasted grid points. Since $4m^2 \geq n$, the total number of grid points is at least $n$. The space efficiency is given by:
    \[
    \frac{n}{4(m+1)^2} \geq \frac{n}{4(2 + \sqrt{n}/2)^2} = \frac{1}{1 + \frac{8}{\sqrt{n}} + \frac{16}{n}},
    \]
    which approaches 1 as $n$ increases.
\end{proof}

\subsubsection{Practical Implications}

By compactly folding the linear architecture into a grid, we adapt the efficient movement strategies for linear DPQA architectures to a practical two-dimensional layout suitable for NAQC architectures. This ensures that the transformation remains optimal in terms of movement counts while minimizing the physical space required for qubit placement.

\section{Evaluation and Comparison} \label{sec:eva}

In this section, we evaluate the performance of our proposed methods (called Linear and Zigzag Paths, respectively) for transforming QFT circuits on NAQC architectures. We compare our approaches with state-of-the-art compilation strategies, focusing on movement counts, cumulative movement distances, and overall fidelity. Our evaluation aims to demonstrate the effectiveness of our methods in minimizing resource overheads while maintaining high-fidelity quantum operations.

We implement Linear and Zigzag Paths in Python and employ Enola's CodeGen framework \cite{tan2024compilation} to generate gate execution and schedule movements. All experiments were conducted on a system running Ubuntu 22.04, equipped with a 40-core Intel Xeon Gold 5215 processor at 2.50 GHz and 512 GB of RAM.

\subsection{The Compared State-of-the-art Compilers}
We compare our methods with the following state-of-the-art compilers for NAQC:
\begin{itemize} 
    \item \textbf{Enola \cite{tan2024compilation}}: A leading compiler for DPQA architectures. Enola adapts qubit connections dynamically during computation and schedules the circuit using a minimal number of Rydberg stages, such that \(CZ\) gates in each stage are executed in parallel using a global Rydberg laser. [GitHub Repository: \href{https://github.com/UCLA-VAST/Enola}{https://github.com/UCLA-VAST/Enola}]
    
    \item \textbf{Atomique \cite{wang2024atomique}}: Another leading compiler for DPQA architectures. Atomique employs multiple AOD arrays to enable parallel movements and enhance movement efficiency. This approach eliminates the need for atom transfers but incurs the cost of extra SWAP gates when qubits in the same array interact. Our setup for Atomique uses two AOD arrays. [Zenodo record: \href{https://zenodo.org/records/10995324}{https://zenodo.org/records/10995324}] 

    \item \textbf{DasAtom \cite{huang2024dasatom}}: A compiler for NAQC architectures where \(CZ\) gates are executed \emph{individually} using local Rydberg lasers. DasAtom exploits long-range interactions to reduce the need for atom movements and has demonstrated excellent performance in improving the overall fidelity. Our setup for DasAtom assumes an atom distance of 3 $\mu$m and an interaction radius of 6 $\mu$m. [GitHub Repository: \href{https://github.com/Huangyunqi/DasAtom}{https://github.com/Huangyunqi/DasAtom}]
\end{itemize}

\subsection{Key Performance Metrics}
We evaluate these methods based on three critical metrics: 
\begin{enumerate}[label=(\arabic*)]
    \item \textbf{Big Move Counts}: The number of significant positional shifts of qubits during circuit execution. Minimising this reduces operational overhead and potential errors.
    
    \item \textbf{Cumulative Movement Distances}: The total distance qubits travel, including both big moves and short offsets executed to resolve placement conflicts (i.e., preventing qubits from occupying the same physical location).
    
    \item \textbf{Overall Fidelity}:  Calculated based on errors from two-qubit gate error, global Rydberg excitation error, atom transfer error, and decoherence error, following the fidelity model outlined in ~\cite{tan2024compilation}: 
\begin{align}
\label{eq:error}
    f = \;\underbrace{\left(f_{2}\right)^{g_{2}}}_{\text{two-qubit gate}} 
    \times \underbrace{ \left(f_{\text{exc}}\right)^{|Q|S - 2g_{2}}}_{\text{Rydberg excitation}}
    % \notag \\&
    \;\times\; \underbrace{\left(f_{\text{trans}}\right)^{N_{\text{trans}}}}_{\text{atom transfer}} 
    \;\times\; \underbrace{\prod_{q \in Q} \left(1 - \frac{T_q}{T_2}\right)}_{\text{decoherence}},
\end{align}
where:
    \begin{itemize}
        \item $f_{2}$ is the two-qubit gate fidelity.
        \item $g_2$ is the number of \(CZ\) gates in the compiled circuit. Atomique may insert SWAP gates and each SWAP gate incurs three $CZ$ gates. 
        \item $Q$ is the qubit set in the compiled circuit. Atomique may use ancilla qubits. 
        \item $S$ is the number of Rydberg stages. Atomique may use more stages than the other compilers, which all use a minimal number of stages.
        \item $f_{\text{exc}}$ is the fidelity of an isolated (i.e., not interacting with another qubit) qubit in a Rydberg stage. For DasAtom, $f_{\text{exc}}=1$ as it employs individually addressible Rydberg lasers.
        \item $f_{\text{trans}}$ measures fidelity losses from atom transfers.
        \item $N_{\text{trans}}$ is the number of atom transfers.  Atomique uses no atom transfers.
        \item $T_q$ is the idling time for qubit $q$.
        \item $T_2$ is the coherence time.
    \end{itemize}
\end{enumerate}

Table~\ref{tab:para} lists the key parameters adopted in our experiments, which are also assumed in \cite{tan2024compilation}.

\begin{table}[hbt]
    \centering
    \caption{Key DPQA/NAQC parameters, where $d$ represents the default spacing (i.e., unit distance) in the SLM array, and $a$ denotes the acceleration of qubit movement.}
    \label{tab:para}
    \renewcommand{\arraystretch}{1.5} % Adjust row height for better spacing
    \setlength{\tabcolsep}{4pt} % Adjust column separation for better spacing
    \begin{tabular}{lccccccc}
        \hline
        \textbf{Parameter} & $f_2$ & $f_{\text{exc}}$ & $f_{\text{trans}}$ & $T_2$ & $T_{\text{trans}}$ & $d$ & $a$ \\
        \hline
        \textbf{Value} & $99.5\%$ & $99.75\%$ & $99.9\%$ & $1.5\,\text{s}$ & $1.5\,\mu\text{s}$ & $15\,\mu\text{m}$ &  $2750\,\text{m/s}^2$ \\
        \hline
    \end{tabular}

\end{table}

\smallskip
\noindent
\textbf{Theoretical Lower Bounds (TLBs)} for movement counts serve as a baseline metric for optimal performance and are calculated as $4(2n-3)$. Note that the ``meet-interact-separate'' action of the last $CZ$ layer (corresponding to the last SWAP gate in Fig.~\ref{fig:maslov}) does not need to swap the two atoms  in practice, which could reduce the TLB by 1.
%\cyan{In fact, the Enola code generation framework automatically eliminates separate movements in the final layer. As a result, the TLB is reduced to $4(2n-3)-1$. }

\subsection{Evaluations}
\label{sec:evaluations}

We focus on QFT circuits as the primary workload, which exhibit dense connectivity. Later, we extend the analysis to sparser circuits (QAOA MaxCut).

\begin{figure}[htbp]
    \begin{minipage}[c]{0.48\columnwidth}
        \centering
        \resizebox{!}{30ex}{% \documentclass{standalone}
% \usepackage{tikz}
% \usepackage{siunitx}
% \usepackage{pgfplots}
% \pgfplotsset{compat=newest}
% \usepgfplotslibrary{units}
% \pgfdeclareplotmark{circle*}{%
%         \node[fill=none,draw=none,shape=circle,inner sep=1pt] {*};  % Circle with an embedded *
%     }
%     \definecolor{idealcolor}{RGB}{31,119,180} % Blue
%     \definecolor{LNNcolor}{RGB}{255,127,14} % Orange
%     \definecolor{ZZcolor}{RGB}{44,160,44} % Green
%     \definecolor{EnolaDycolor}{RGB}{214,39,40} % Red
%     \definecolor{EnolaScolor}{RGB}{148,103,189} % Purple
% \begin{document}
    \begin{tikzpicture}
        \begin{axis}[
            width=12cm,
            height=8cm,
            % X-axis label position control
            xlabel={Qubit number in QFT circuit},
            xlabel style={
                at={(axis description cs:0.5,-0.1)},  % Move label down
                anchor=north,
                font=\Large
            },
            % Y-axis label position control
            ylabel={Big Move Counts},
            ylabel style={
                at={(axis description cs:-0.12,0.5)},  % Move label left
                anchor=south,
                rotate=0,
                font=\Large
            },
            % Axis ranges
            % ymode = log,
            % log basis y = 10,
            xmin=0, xmax=55,
            xtick={5,10,...,50},
            grid=major,
            % Legend position control
            legend style={at={(0.05,1)}, anchor=north west},
            cycle list name=color list,
            mark size=1.5pt,
            line width=1pt
        ]
        
        % Theoretical lower bound (solid line)
        \addplot[
            color=idealcolor,
            thick,
            mark=none,
            domain=2:55
        ] {8*x - 13};
        \addlegendentry{Theoretical Lower Bound}
        
        % LNN Path
        \addplot[LNNPlot] coordinates {
            (5, 27) (6, 35) (7, 43) (8, 51) (9, 59) (10, 67) (11, 75) (12, 83) (13, 91) (14, 99) (15, 107) (16, 115) (17, 123) (18, 131) (19, 139) (20, 147) (21, 155) (22, 163) (23, 171) (24, 179) (25, 187) (26, 195) (27, 203) (28, 211) (29, 219) (30, 227) (31, 235) (32, 243) (33, 251) (34, 259) (35, 267) (36, 275) (37, 283) (38, 291) (39, 299) (40, 307) (41, 315) (42, 323) (43, 331) (44, 339) (45, 347) (46, 355) (47, 363) (48, 371) (49, 379) (50, 387)
        };
        \addlegendentry{Linear Path}
        
        % Zigzag Folded Path
        \addplot[ZZPlot] coordinates {
            (5, 27) (6, 35) (7, 43) (8, 51) (9, 59) (10, 67) (11, 75) (12, 83) (13, 91) (14, 99) (15, 107) (16, 115) (17, 123) (18, 131) (19, 139) (20, 147) (21, 155) (22, 163) (23, 171) (24, 179) (25, 187) (26, 195) (27, 203) (28, 211) (29, 219) (30, 227) (31, 235) (32, 243) (33, 251) (34, 259) (35, 267) (36, 275) (37, 283) (38, 291) (39, 299) (40, 307) (41, 315) (42, 323) (43, 331) (44, 339) (45, 347) (46, 355) (47, 363) (48, 371) (49, 379) (50, 387)
        };
        \addlegendentry{Zigzag Path}
        % Enola Dynamic
        \addplot[EnolaDyPlot] coordinates {
            (5,30) (6,48) (7,57) (8,76) (9,99) (10,125) (11,131) (12,160) (13,197) (14,233) (15,255) (16,276) (17,307) (18,335) (19,379) (20,416) (21,447) (22,481) (23,527) (24,567) (25,598) (26,631) (27,669) (28,709) (29,753) (30,813) (31,858) (32,882) (33,924) (34,981) (35,1057) (36,1083) (37,1126) (38,1197) (39,1260) (40,1352) (41,1356) (42,1425) (43,1472) (44,1527) (45,1569) (46,1661) (47,1700) (48,1802) (49,1871) (50,1922)
        };
        \addlegendentry{Enola}
        
        % Enola Static
        % \addplot[EnolaSPlot] coordinates {
        %     (5,32) (6,52) (7,76) (8,96) (9,116) (10,156) (11,180) (12,192) (13,264) (14,288) (15,340) (16,364) (17,400) (18,476) (19,496) (20,536) (21,588) (22,684) (23,744) (24,776) (25,816) (26,892) (27,940) (28,1012) (29,1100) (30,1140) (31,1216) (32,1288) (33,1368) (34,1452) (35,1548) (36,1556) (37,1632) (38,1752) (39,1792) (40,1940) (41,2016) (42,2108) (43,2180) (44,2248) (45,2308) (46,2520) (47,2552) (48,2624) (49,2708) (50,2980)
        % };
        % \addlegendentry{Enola Static}

        \addplot [DasAtomPlot]
        coordinates {%
        ( 5 , 0) ( 6 , 3) ( 7 , 3) ( 8 , 3) ( 9 , 10) ( 10 , 12) ( 11 , 12) ( 12 , 19) ( 13 , 21) ( 14 , 24) ( 15 , 31) ( 16 , 36) ( 17 , 34) ( 18 , 39) ( 19 , 46) ( 20 , 55) ( 21 , 52) ( 22 , 60) ( 23 , 71) ( 24 , 81) ( 25 , 86) ( 26 , 84) ( 27 , 85) ( 28 , 92) ( 29 , 95) ( 30 , 106) ( 31 , 108) ( 32 , 116) ( 33 , 124) ( 34 , 123) ( 35 , 133) ( 36 , 148) ( 37 , 144) ( 38 , 161) ( 39 , 160) ( 40 , 173) ( 41 , 183) ( 42 , 186) ( 43 , 191) ( 44 , 205) ( 45 , 219) ( 46 , 220) ( 47 , 242) ( 48 , 239) ( 49 , 265) ( 50 , 230) };
        \addlegendentry{DasAtom}

        \addplot [AtomiquePlot]
        coordinates {%
        (5, 4) (6, 11) (7, 16) (8, 23) (9, 24) (10, 43) (11, 56) (12, 48) (13, 85) (14, 80) (15, 109) (16, 132) (17, 132) (18, 137) (19, 161) (20, 162) (21, 245) (22, 264) (23, 249) (24, 249) (25, 245) (26, 246) (27, 317) (28, 369) (29, 336) (30, 237) (31, 293) (32, 355) (33, 363) (34, 376) (35, 379) (36, 458) (37, 504) (38, 473) (39, 440) (40, 496) (41, 560) (42, 582) (43, 622) (44, 496) (45, 539) (46, 881) (47, 878) (48, 854) (49, 739) (50, 819) };
        \addlegendentry{Atomique}
        \end{axis}
    \end{tikzpicture}
% \end{document}
}
        \subcaption{}
        \label{fig:movement_counts}
    \end{minipage}
    \begin{minipage}[c]{0.48\columnwidth}
        \centering
        \resizebox{!}{30ex}{% \documentclass{standalone}
% \usepackage{tikz}
% \usepackage{siunitx}
% \usepackage{pgfplots}
% \pgfplotsset{compat=newest}
% \usepgfplotslibrary{units}
% \pgfdeclareplotmark{circle*}{%
%         \node[fill=none,draw=none,shape=circle,inner sep=1pt] {*};  % Circle with an embedded *
%     }
%     \definecolor{idealcolor}{RGB}{31,119,180} % Blue
%     \definecolor{LNNcolor}{RGB}{255,127,14} % Orange
%     \definecolor{ZZcolor}{RGB}{44,160,44} % Green
%     \definecolor{EnolaDycolor}{RGB}{214,39,40} % Red
%     \definecolor{EnolaScolor}{RGB}{148,103,189} % Purple
% \begin{document}
    \begin{tikzpicture}

        \begin{axis}[
            width=12cm,
            height=8cm,
            xlabel={Qubit number in QFT circuit},
            xlabel style={
                at={(axis description cs:0.5,-0.1)},  % Move label down
                anchor=north,
                font=\Large
            },
            ylabel={\shortstack{Cumulative Movement Distances \\ (Grid Units)}},
            ylabel style={
                at={(axis description cs:-0.12,0.5)},  % Move label left
                anchor=south,
                rotate=0,
                font=\large
            },
            ymode = log,
            log basis y = 10,
            xmin = 0,
            xmax = 55,
            xtick={5,10,...,50},
            grid=major,
            legend style={at={(0.95,0)}, anchor=south east},
            cycle list name=color list,
            mark size=1.5pt,
            line width=0.5pt
        ]
        
        % Theoretical lower bound (solid line)
        \addplot[
            idealPlot,
            domain=2:55
        ] {8*x - 13};
        \addlegendentry{Theoretical Lower Bound}
        
        \addplot[LNNPlot] coordinates {(5 , 36.6842105263158) (6 , 47.63157894736848) (7 , 58.57894736842115) (8 , 69.52631578947383) (9 , 80.47368421052651) (10 , 91.4210526315792) (11 , 102.36842105263187) (12 , 113.31578947368456) (13 , 124.26315789473726) (14 , 135.21052631578996) (15 , 146.15789473684268) (16 , 157.1052631578954) (17 , 168.0526315789481) (18 , 179.00000000000082) (19 , 189.94736842105357) (20 , 200.89473684210623) (21 , 211.84210526315897) (22 , 222.78947368421169) (23 , 233.7368421052644) (24 , 244.68421052631712) (25 , 255.63157894736983) (26 , 266.57894736842235) (27 , 277.52631578947484) (28 , 288.4736842105273) (29 , 299.4210526315798) (30 , 310.3684210526323) (31 , 321.3157894736848) (32 , 332.2631578947373) (33 , 343.21052631578976) (34 , 354.15789473684225) (35 , 365.10526315789474) (36 , 376.0526315789472) (37 , 386.9999999999997) (38 , 397.9473684210522) (39 , 408.8947368421047) (40 , 419.8421052631572) (41 , 430.78947368420967) (42 , 441.73684210526216) (43 , 452.6842105263146) (44 , 463.63157894736713) (45 , 474.5789473684196) (46 , 485.5263157894721) (47 , 496.4736842105246) (48 , 507.4210526315771) (49 , 518.3684210526292) (50 , 529.315789473681)};
        \addlegendentry{Linear Path}
        
        \addplot[ZZPlot] coordinates {(5 , 39.63157894736844) (6 , 52.05263157894745) (7 , 64.47368421052646) (8 , 76.89473684210547) (9 , 89.3157894736845) (10 , 101.73684210526352) (11 , 114.15789473684256) (12 , 126.5789473684216) (13 , 139.00000000000065) (14 , 151.42105263157973) (15 , 163.8421052631588) (16 , 176.26315789473787) (17 , 188.68421052631692) (18 , 201.105263157896) (19 , 213.5263157894751) (20 , 225.94736842105416) (21 , 238.36842105263324) (22 , 255.4210526315809) (23 , 272.68421052631754) (24 , 291.63157894737) (25 , 310.5789473684224) (26 , 329.52631578947484) (27 , 348.47368421052727) (28 , 367.4210526315797) (29 , 386.368421052632) (30 , 405.3157894736845) (31 , 424.263157894737) (32 , 443.2105263157895) (33 , 462.157894736842) (34 , 481.10526315789446) (35 , 500.0526315789469) (36 , 518.999999999999) (37 , 537.94736842105) (38 , 556.8947368421009) (39 , 575.8421052631519) (40 , 594.7894736842027) (41 , 614.1578947368323) (42 , 637.3157894736722) (43 , 661.1052631578784) (44 , 686.5789473684002) (45 , 712.052631578922) (46 , 737.5263157894434) (47 , 762.9999999999653) (48 , 788.4736842104869) (49 , 813.9473684210083) (50 , 839.4210526315297)};
        \addlegendentry{Zigzag Path}
        
        \addplot[EnolaDyPlot] coordinates {(5,43.789473684210556) (6,77.10526315789492) (7,90.94736842105281) (8,125.52631578947407) (9,165.78947368421123) (10,220.7368421052645) (11,231.36842105263318) (12,289.05263157894865) (13,360.15789473684276) (14,445.94736842105215) (15,487.31578947368325) (16,548.315789473681) (17,627.1052631578826) (18,685.1052631578788) (19,790.9999999999744) (20,885.1052631578629) (21,962.1052631578557) (22,1019.9999999999544) (23,1153.578947368374) (24,1257.2105263157398) (25,1371.4736842104799) (26,1433.999999999954) (27,1570.421052631531) (28,1660.1052631578507) (29,1820.7894736841672) (30,1923.2105263157432) (31,2117.9999999999454) (32,2211.315789473611) (33,2349.3684210525366) (34,2492.6842105261994) (35,2725.8421052629965) (36,2855.631578947188) (37,2901.6842105261194) (38,3158.5789473681834) (39,3417.3684210523497) (40,3817.315789473335) (41,3805.631578947002) (42,4010.684210525929) (43,4235.421052631243) (44,4443.68421052607) (45,4559.736842105097) (46,4880.000000000024) (47,5034.842105263283) (48,5293.894736842361) (49,5775.789473684695) (50,5961.631578948)};
        \addlegendentry{Enola}
        
        % \addplot[EnolaSPlot] coordinates {
        % (5,56.52631578947374) (6,95.36842105263172) (7,140.63157894736872) (8,187.6842105263164) (9,234.0000000000009) (10,329.89473684210566) (11,405.0526315789474) (12,454.84210526315724) (13,607.2631578947316) (14,691.8947368420968) (15,837.0526315789299) (16,936.2105263157717) (17,1092.5263157894478) (18,1306.315789473659) (19,1430.9473684210275) (20,1563.7894736841863) (21,1718.7368421052372) (22,2048.947368421029) (23,2252.4210526315414) (24,2450.210526315735) (25,2657.5789473683517) (26,2908.631578947275) (27,3167.6842105262112) (28,3372.947368420928) (29,3867.1578947366697) (30,4023.4736842103443) (31,4406.631578947241) (32,4627.263157894679) (33,5093.7894736842545) (34,5514.00000000014) (35,5835.473684210779) (36,6076.000000000322) (37,6466.526315789835) (38,7016.947368421573) (39,7237.263157895295) (40,7843.36842105341) (41,8396.947368421885) (42,8778.94736842197) (43,9293.15789473776) (44,9621.894736843096) (45,10029.684210527366) (46,10914.000000001082) (47,11341.263157895957) (48,11846.63157894856) (49,12330.736842106515) (50,13410.736842106635)
        % };
        % \addlegendentry{Enola static}

        \addplot [DasAtomPlot]
        coordinates {%
        ( 5 , 0.0) ( 6 , 5.789473684210526) ( 7 , 5.421052631578947) ( 8 , 5.789473684210527) ( 9 , 18.052631578947363) ( 10 , 26.263157894736835) ( 11 , 28.1578947368421) ( 12 , 47.10526315789478) ( 13 , 50.78947368421058) ( 14 , 56.73684210526322) ( 15 , 75.05263157894753) ( 16 , 89.00000000000024) ( 17 , 91.73684210526343) ( 18 , 105.52631578947404) ( 19 , 124.15789473684254) ( 20 , 149.42105263157956) ( 21 , 144.73684210526372) ( 22 , 168.57894736842184) ( 23 , 206.8947368421061) ( 24 , 235.00000000000105) ( 25 , 261.3684210526327) ( 26 , 240.6315789473697) ( 27 , 238.0526315789488) ( 28 , 267.63157894736963) ( 29 , 270.15789473684345) ( 30 , 311.6315789473695) ( 31 , 312.2105263157905) ( 32 , 338.157894736843) ( 33 , 364.1578947368425) ( 34 , 370.4210526315796) ( 35 , 400.2105263157897) ( 36 , 452.3684210526312) ( 37 , 465.3157894736834) ( 38 , 502.42105263157777) ( 39 , 510.21052631578783) ( 40 , 538.6842105263129) ( 41 , 581.3157894736788) ( 42 , 578.1052631578889) ( 43 , 625.1578947368337) ( 44 , 659.4736842105139) ( 45 , 714.4736842105108) ( 46 , 716.2631578947185) ( 47 , 823.5263157894489) ( 48 , 785.4736842105008) ( 49 , 922.7368421052312) ( 50 , 777.5789473683996) };
        \addlegendentry{DasAtom}

        \addplot [AtomiquePlot]
        coordinates {%
        (5, 8.485281374238571) (6, 22.62741699796952) (7, 33.94112549695429) (8, 46.66904755831214) (9, 63.63961030678928) (10, 120.2081528017131) (11, 137.1787155501902) (12, 144.2497833620557) (13, 311.1269837220809) (14, 217.7888886054566) (15, 408.7077195258245) (16, 543.0580079512685) (17, 704.2783540618013) (18, 440.8204178980326) (19, 804.6875169902911) (20, 641.0646203948496) (21, 1169.922780631882) (22, 1295.909782382184) (23, 1274.249572225166) (24, 1483.133501246079) (25, 1095.034314086523) (26, 1266.914329801313) (27, 1454.445162982603) (28, 1863.966847759885) (29, 1641.236972429958) (30, 805.4487768290221) (31, 946.9955996617798) (32, 1067.253016881391) (33, 1926.28762911958) (34, 2064.889946205207) (35, 2018.067297058727) (36, 2730.840941253017) (37, 3109.686163329204) (38, 1865.068314237348) (39, 2061.304090523308) (40, 2672.13551277076) (41, 3046.909369690031) (42, 3067.204668844678) (43, 3203.166924883476) (44, 2729.503295849729) (45, 2900.648872898975) (46, 4280.724466660904) (47, 4572.311623471759) (48, 4456.916083837345) (49, 2892.719922863976) (50, 4000.288006637702) };
        \addlegendentry{Atomique}
        \end{axis}
    \end{tikzpicture}
% \end{document}
}
        \subcaption{}
        \label{fig:movement_distances}
    \end{minipage}
    
    \begin{minipage}[c]{0.48\columnwidth}
        \centering
        \resizebox{!}{30ex}{\input{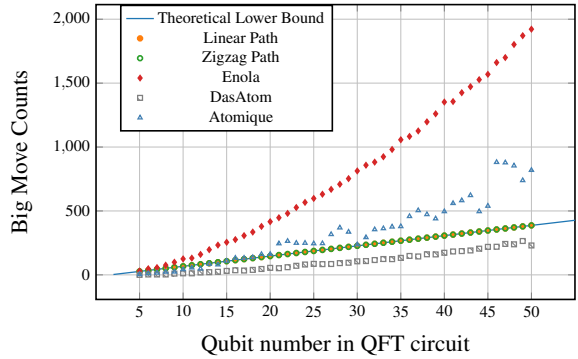}}
        \subcaption{}
        \label{fig:totoal-fid}
    \end{minipage}
    \begin{minipage}[c]{0.48\columnwidth}
        \centering
        \resizebox{!}{32ex}{\input{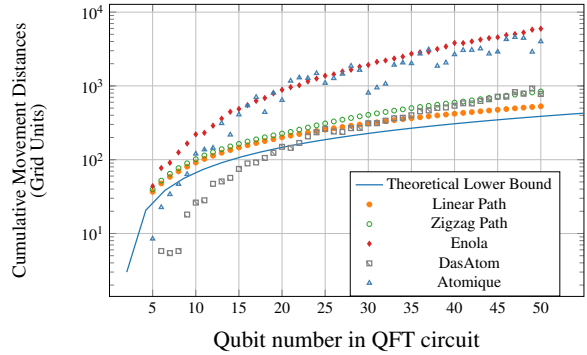}
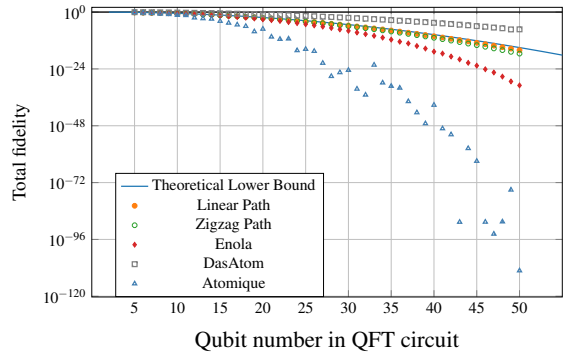
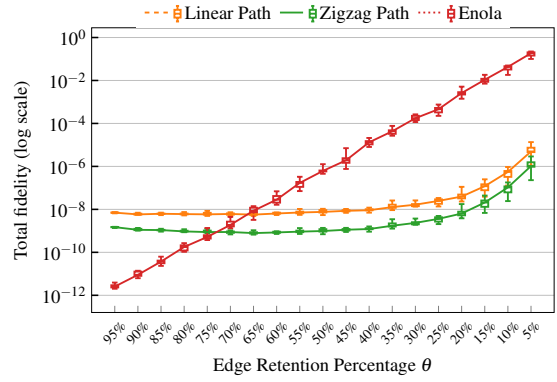}
        \subcaption{}
        \label{fig:qaoa}
    \end{minipage}
    \caption{
    (\subref{fig:movement_counts}) Big move counts for QFT circuits as a function of qubit count $n$, where the TLB is given by $4(2n-3)$.
    (\subref{fig:movement_distances}) Cumulative movement distances for different methods in QFT circuits as the qubit count $n$ increases, with the TLB given by $4(2n-3)d$.
    (\subref{fig:totoal-fid}) Total fidelity comparison for different methods across varying $n$.
    (\subref{fig:qaoa}) Total circuit fidelity for random MaxCut QAOA circuits as a function of the edge retention percentage \(\theta\). For each value of \(\theta\), 10 random circuits were generated and evaluated.
    }
\end{figure}
\subsubsection{Big Move Counts}

Fig.~\ref{fig:movement_counts} illustrates the big move counts for each compiler on QFT circuits of increasing qubit sizes.
\begin{itemize}
    \item \textbf{Linear \& Zigzag Paths} achieve TLB for big move counts, demonstrating their effectiveness in minimizing major atom movements.
    \item \textbf{Enola}'s compilation strategy results in exponentially higher big move counts than the other methods, due to its necessity to update atom placement after each Rydberg stage.  
    \item \textbf{Atomique} uses SWAP gates to reduce  atom movements, but its overall big move counts are higher than Linear and Zigzag Paths when $n\geq 15$.
    \item \textbf{DasAtom} achieves fewer big moves than the TLB by leveraging long-range interactions, which eliminate the need for certain atom movements. This is contingent upon the device supporting individually addressable Rydberg lasers.
\end{itemize}

\subsubsection{Cumulative Movement Distances}

The Cumulative Movement Distance metric quantifies the total distance travelled by atoms during the execution of QFT circuits. Fig.~\ref{fig:movement_distances} illustrates the cumulative movement distances, measured in grid units, for each method as a function of the number of qubits. Notably, each big move in the TLB covers a unit distance. In our methods, big moves also approximates a unit distance, accounting for offset moves. However, in the compared compilers, a single big move may cover much larger distance. It is worth stressing that DasAtom uses a much smaller unit distance ($3\mu$m instead of $15\mu$m).

% \magenta{Data: big move counts, distance, offset counts, distance} \cyan{The experimental distance data is unavailable due to the server being offline during the holiday shutdown.}
\begin{itemize}
    \item \textbf{Linear \& Zigzag Paths} maintain minimal cumulative movement distances, closely approximating the TLB. While generally efficient, the Zigzag path shows a slight increase in cumulative movement distance as the qubit count grows, due to the need for additional offsets in the grid architecture.
    \item \textbf{Enola} exhibits the highest cumulative movement distances among all methods, like big move counts.
    \item \textbf{Atomique} demonstrates efficiency for smaller circuits (less than 10 qubits) but experiences an exponential decrease in performance with increasing qubit count.
    \item \textbf{DasAtom} achieves moderate cumulative movement distances and outperforms TLB for QFT circuits with up to approximately 20 qubits.
\end{itemize}

% \usepackage{multirow}
% \begin{table}[]
% \resizebox{\columnwidth}{!}{
% \begin{tabular}{l|lll|lll|lll}
% \multirow{2}{*}{file\_name} & \multicolumn{3}{c}{LNN}                    & \multicolumn{3}{c}{ZZ}                     & \multicolumn{3}{c}{Enola}                  \\
%                             & total & bigmov & big/total & total & bigmov & big/total & total & bigmov & big/total \\
% qft 10                      & 91.42105         & 66.84211    & 0.731146  & 101.7368         & 66.84211    & 0.65701   & 220.7368         & 147.7895    & 0.669528  \\
% qft 20                      & 200.8947         & 146.8421    & 0.730941  & 225.9474         & 146.8421    & 0.649895  & 885.1053         & 580.6316    & 0.656003  \\
% qft 30                      & 310.3684         & 226.8421    & 0.73088   & 405.3158         & 226.8421    & 0.559668  & 1923.211         & 1253.368    & 0.651572  \\
% qft 40                      & 419.8421         & 306.8421    & 0.730851  & 594.7895         & 306.8421    & 0.515884  & 3817.316         & 2589.947    & 0.678473  \\
% qft 50                      & 529.3158         & 386.8421    & 0.730834  & 839.4211         & 386.8421    & 0.460844  & 5961.632         & 4050.789    & 0.679477 
% \end{tabular}
% }
% \end{table}

\subsubsection{Overall Fidelity Analysis}

% \begin{figure}
%     \centering
%     \resizebox{\columnwidth}{!}{\input{figure/data_total_fid}}
%     \caption{Total fidelity comparison for different methods across varying $n$.}
%     \label{fig:totoal-fid}
% \end{figure}

Fig.~\ref{fig:totoal-fid} compares the overall fidelity achieved by each method, where the TLB is calculated based on four atom transfers per CZ gate.

%\cyan{where TLB is calculated according to the following relationships:
%\begin{align*}
%& Q= n, \\
%& g_{2}= n(n - 1), \\
%& S = 2(2n - 3), \\
%& N_{\text{trans}} = 4n(n - 1), \\
%& T_q = T_{\text{movd}} \,(n - 1) + T_{\text{trans}} \, 2(n - 1).
%\end{align*}
% $f_{2}^{n(n-1)} \cdot f_{\text{exc}}^{n\times 4(n-3) - 2*(x*(x-1)))}  f_{\text{trans}}^{(4n(n-1)))} (1 - \frac{T_{movd}\cdot(n-1) + T_{trans}\cdot2(n-1)} {T_{2}})^{n}$

\begin{itemize}
    \item \textbf{Linear \& Zigzag Paths} achieve the highest fidelity, as their movement efficiency minimises errors from atom transfers and decoherence.
    \item \textbf{Enola}'s methods show exponentially lower fidelity, reflecting their higher movement counts and distances.
    \item \textbf{Atomique} experiences fidelity degradation as circuit size increases, primarily due to its reliance on SWAP gates and ancilla qubits.
    \item \textbf{DasAtom} achieves higher overall fidelity than the TLB by leveraging individually addressable Rydberg lasers and long-range interactions. Since it employs local two-qubit operations without requiring global Rydberg illumination, it avoids global Rydberg excitation error. Effectively, this is equivalent to setting $f_\text{exc}=1$ for DasAtom in \eqref{eq:error}.
\end{itemize}

Since the first three terms in \eqref{eq:error} are identical for TLB, Line and Zigzag Paths, and Enola, the overall fidelity for these compilers is predominantly determined by big move counts and cumulative movement distances.

\smallskip
\noindent
\textbf{Quantifying Fidelity Gaps.}  
Table~\ref{tab:qft_comparison} highlights the stark fidelity gaps in larger QFT circuits. For $n=30$, Enola's fidelity is more than $180\times$ lower than our Linear Path, while Atomique's is over $10^{18}\times$ lower. At $n=50$, these gaps widen dramatically to factors of $10^{15}$ and $10^{94}$, respectively, emphasizing the critical role of movement minimization in large-scale circuits.

Meanwhile, the table suggests that with individually addressable Rydberg lasers, DasAtom's fidelity could surpass the TLB by more than  $64\times$ for QFT-30 and $10^6\times$ for QFT-50. This advantage persists even when excitation error is ignored for the TLB, where DasAtom's fidelity remains nearly $10\times$ higher for QFT-30 and over $200\times$ higher for QFT-50. This improvement stems from DasAtom's use of long-range interactions. However, we note that current state-of-the-art NAQC with individually addressed Rydberg lasers faces limitations in two-qubit gate fidelity \cite{graham_multi-qubit_2022}.

Finally, the table reveals that Enola and Atomique exhibit different performance behaviours depending on whether $f_\text{exc}=0.9975$ or $f_\text{exc}=1$. In the first case, Enola outperforms Atomique by factors of $10^{16}\times$ for QFT-30 and $10^{78}\times$ for QFT-50. However, when $f_\text{exc}=1$, Atomique surpasses Enola by $18.9\times$ for QFT-30 and $10^9\times$ for QFT-50, which is consistent with the experimental findings reported in \cite{huang2024dasatom}.

\begin{table}[ht]
    \centering
    \caption{Comparison on QFT-30 and QFT-50, where distances are measured in unit grid distance ($d=3\mu$m for DasAtom and $d=15\mu$m for all the other compilers). In Atomique, this distance is not computed by the code. Note that excitation error does not impact DasAtom's fidelity, as it utilizes local Rydberg lasers.}
    \label{tab:qft_comparison}
    \resizebox{\columnwidth}{!}{%
        \begin{tabular}{|l|l|c|c|c|r|c|c|c|c|c|c|}
    \toprule
    \multirow{2}{*}{\textbf{Circuit}} & \multirow{2}{*}{\textbf{Method}} & \multirow{2}{*}{$g_2$} & \multirow{2}{*}{$|Q|$} & \multirow{2}{*}{$S$} & \multirow{2}{*}{$N_{\text{trans}}$} & \multicolumn{2}{c|}{\textbf{Big Move}} & \multicolumn{2}{c|}{\textbf{Offset Move}} & \multicolumn{2}{c|}{\textbf{Overall Fidelity}} \\
    & & & & & & count & distance & count & distance & $f_{\text{exc}}=0.9975$ & $f_{\text{exc}}=1$ \\ 
    \midrule
    \multirow{6}{*}{QFT-30} 
     & TLB     & 870 & 30 & 114  & 3478 &  227 & 227 & 0 & 0 & $5.7\times10^{-5}$ & $3.8\times10^{-4}$\\
     & Linear  & 870 & 30 & 114  & 3478 &  227 & 226.84 & 681 & 83.53 & $2.1\times10^{-6}$& $1.4\times10^{-5}$\\
     & Zigzag  & 870 & 30 & 114  & 3478 &  227 & 226.84 & 1413 & 178.47 & $1.7\times10^{-6}$& $1.1\times10^{-5}$\\
     & Enola   & 870 & 30 & 114  & 3478 &  813 & 1253.37 & 4681 & 669.84 & $1.2\times10^{-8}$ & $7.9\times 10^{-7}$\\
     & Atomique & 954 & 43 & 466  & 0    &  237 & 805.45 & 229 & \textemdash & $2.9\times10^{-25}$ & $1.5\times10^{-5}$ \\
     & DasAtom & 870 & 30 & 786  & 460  &  106 & 243.21 & 640 & 68.42 & $3.7\times10^{-3}$ & $3.7\times10^{-3}$ \\
    \midrule
    \multirow{6}{*}{QFT-50} 
     & TLB     & 2450 & 50 & 194  & 9798 & 387 & 387 & 0 & 0 & $1.4\times10^{-15}$ & $2.4\times10^{-10}$\\
     & Linear  & 2450 & 50 & 194  & 9798 & 387 & 386.84 & 1161 & 142.47 & $8.3\times10^{-17}$& $1.4\times10^{-11}$\\
     & Zigzag  & 2450 & 50 & 194  & 9798 & 387 & 386.84 & 3501 & 452.58 & $2.9\times10^{-18}$& $5.0\times10^{-13}$\\
     & Enola   & 2450 & 50 & 194  & 9798 & 1922 & 4050.79 & 13318 & 1910.84 & $1.1\times10^{-31}$ & $1.8\times 10^{-26}$\\
     & Atomique & 2669 & 79 & 1145 & 0    & 819 & 4000.29 & 326 & \textemdash & $6.6\times10^{-110}$ & $2.2\times10^{-17}$ \\ 
     & DasAtom & 2450 & 50 & 1554 & 1264 & 221 & 567.16 & 1701 & 182.05  & $5.2\times10^{-8}$ & $5.2\times10^{-8}$ \\
    \bottomrule
\end{tabular}
    }
\end{table}
\subsection{Extension to Circuits with Fewer Interactions}
\label{sec:sparser}

While the QFT circuit is characterized by a fully connected interaction graph, many practical quantum circuits do not require full connectivity. To evaluate the adaptability of our methods, we extend our analysis to MaxCut QAOA circuits \cite{farhi2014quantum} with reduced connectivity. 

\subsubsection{Procedure for Generating Random MaxCut QAOA Circuits}  
To generate a random $n$-qubit MaxCut QAOA circuit \(C\), the following steps are performed:  

\begin{enumerate}  
    \item \textbf{Initialization:} 
    %Begin with the $n$-qubit QFT circuit, similar to the one shown in Fig.~\ref{fig:qft-5-par}, and replace all \(CP\) gates with \(CZ\) gates. Retain only the first \(n-1\) layers of the circuit and denote the resulting circuit as \(C_{QFT}\). 
    Define a complete graph \(G = (V, E)\), where \(V = \{0, 1, \dots, n-1\}\) represents the qubits, and \(E = \{(i, j) \mid 0 \leq i < j < n\}\) is the set of edges connecting all pairs of qubits. 
    
    \item \textbf{Edge Sampling:} For each edge \((i, j) \in E\), decide whether to include the edge in a subset \(E_\theta \subseteq E\). An edge is included in \(E_\theta\) with a probability of \(\theta \%\), where \(\theta\) is the specified percentage of edges to retain.  

    \item \textbf{Insertion of \(CZ\) Gates:} Create an empty circuit \(C\). For each edge \((i, j) \in E_\theta\), insert a gate \(CZ(i, j)\) into the circuit \(C\).   
\end{enumerate}  
This procedure ensures the systematic generation of random MaxCut QAOA circuits, with the density of the \(CZ\) gates determined by the parameter \(\theta\). The randomness in edge selection allows for the exploration of circuit behaviour across varying levels of sparsity.  

To compile any such random circuit \(C\), we first consider the $n$-qubit circuit \(C_{QFT}\), which is a simplified version of the QFT-$n$ circuit (cf. Fig.~\ref{fig:qft-5-par}) obtained by replacing all \(CP\) gates directly with \(CZ\) gates. Our Linear and Zigzag Paths methods can be directly applied to \(C_{QFT}\). To adapt these methods for $C$, we exploit the fact that the order of \(CZ\) gates in $C$ is unimportant. We arrange the $CZ$ gates of $C$ into layers matching those of \(C_\text{QFT}\), though some \(CZ\) gates potentially missing in certain layers. The compilation of \(C\) then proceeds by simulating the compilation of \(C_\text{QFT}\), but with a modified meet-interact-swap operation. Specifically, for each m-stage $k$: (1) we apply the meet and interact steps to the $CZ$ gates in $C_k$ (the subset of $CZ$ gates in $C$ that appear in the $k$-th m-stage of \(C_\text{QFT}\), denoted by $C'_k$); (2) we apply the meet step to the $CZ$ gates present in the $C'_k$ but not in $C_k$; and (3) we perform the swap step for all $CZ$ gates in $C'_k$. This effectively splits each parallel atom ``meet'' move from the \(C_\text{QFT}\) compilation into two ``meet'' moves for $C$, ensuring consistent qubit mapping transitions with QFT-$n$ (cf. Fig.~\ref{fig:qft5-lnn}).

% \cyan{Our compilation methods can be applied to compile any such random circuit \(C\). We begin with the $n$-qubit circuit \(C_{\text{QFT}}\), which is similar to QFT-$n$ circuit (cf. Fig.~\ref{fig:qft-5-par}) and obtained by replacing all \(CP\) gates with \(CZ\) gates. The interaction graph of \(C_{\text{QFT}}\) remains a complete graph, making it evident that our compilation methods can be directly applied to \(C_{\text{QFT}}\). This steps serves as the initialization step in the overall procedure. }

Fig.~\ref{fig:qaoa} shows that the total circuit fidelity as a function of $\theta$ under our methods and Enola. It suggests that if $\theta\geq 70$, our methods perform better than Enola. Note that Enola also exploits gate commutability for these circuits.

\section{Conclusion} \label{sec:con}
This work introduced optimal compilation strategies for QFT circuits on NAQC platforms. Our proposed Linear and Zigzag Path strategies achieve theoretical lower bounds in atom movement counts while maintaining high circuit fidelity. Comprehensive evaluations demonstrate that these methods exponentially outperform state-of-the-art DPQA approaches in movement efficiency and overall fidelity. This exponential fidelity gap highlights the significant advantage of minimizing atom movements and optimizing qubit routing on NAQC systems, underscoring the critical importance of movement-aware compilation for preserving the overall circuit fidelity in scalable quantum computing.

While our compilation strategies are currently tailored for specific circuits like QFT, they reveal a huge fidelity gap compared to existing compilers, highlighting the potential for substantial improvements in general NAQC compilation. We expect that this work will motivate the development of more advanced NAQC compilers. Furthermore, DasAtom's performance, exceeding even the theoretical lower bound, suggests that architectural innovations such as individually addressable Rydberg lasers or zoned architectures \cite{Bluvstein_2023} are crucial for NAQC scalability.

%\section{Reference}
% \printbibliography
% \bibliography{atom}
% \bibliography{qft}
% \bibliography{ref}
\bibliography{clean}

\section*{Data Availability}
The datasets generated and analyzed during this study are available in the repository on GitHub: \href{https://github.com/gcc-bug/qft_atom}{https://github.com/gcc-bug/qft\_atom}. The final processed data used in this study are included in the data directory.

\section*{Acknowledgements}
Work partially supported by the National Science Foundation of China (12471437) and the Beijing Nova Program (20220484128).

\section*{Author contributions statement}
%Y.L. proposed the initial idea. D.G. conducted the experiments. D.G. and S.Y. wrote the initial draft. S.L. improved the draft and suggested the extension experiments. All authors reviewed and approved the manuscript.
SL conceived the initial research question and proposed the core ideas. DG developed the methodology, performed the experiments, and drafted the manuscript. YL and SY contributed to discussions and all authors contributed to manuscript revisions.

\section*{Additional information}
The authors declare no competing interests.

\end{document}